\documentclass[english]{IEEEtran}
\usepackage[T1]{fontenc}
\usepackage[latin9]{inputenc}
\usepackage{float}
\usepackage{amsmath}
\usepackage{amsthm}
\usepackage{amssymb}
\usepackage{graphicx}

\makeatletter

\floatstyle{ruled}
\newfloat{algorithm}{tbp}{loa}
\providecommand{\algorithmname}{Algorithm}
\floatname{algorithm}{\protect\algorithmname}

  \theoremstyle{plain}
  \newtheorem{lem}{\protect\lemmaname}

\usepackage{amsfonts}
\usepackage{amsthm}
\usepackage{mathtools}
\usepackage{cite}
\usepackage{array}
\usepackage{algorithm}
\usepackage{algorithmic}
\usepackage{subfigure}

\DeclareMathOperator{\tr}{tr}

\DeclareMathOperator{\D}{\mathsf{D}}
\DeclareMathOperator{\myvec}{vec}

\DeclareMathOperator*{\st}{subject~to}
\DeclareMathOperator*{\maxi}{maximize}
\DeclareMathOperator*{\mini}{minimize}

\DeclareMathOperator{\blkdiag}{blkdiag}

\newcommand{\herm}{^{\mbox{\scriptsize H}}}
\newcommand{\trans}{^{\mbox{\scriptsize T}}}
\renewcommand{\Im}{\mathrm{Im}}
\renewcommand{\Re}{\mathrm{Re}}

\makeatother

\makeatother

\usepackage{babel}
  \providecommand{\lemmaname}{Lemma}

\begin{document}

\title{Conic Quadratic Formulations for Wireless Communications Design }

\author{Quang-Doanh Vu, Markku Juntti,~\IEEEmembership{Senior Member, IEEE},
Een-Kee Hong,~\IEEEmembership{Senior Member, IEEE}, and Le-Nam Tran,~\IEEEmembership{Member, IEEE}\thanks{This work was supported in part by the Academy of Finland under project
Message and CSI Sharing for Cellular Interference Management with
Backhaul Constraints (MESIC) belonging to the WiFIUS program with
NSF, and in part by a research grant from Science Foundation Ireland
(SFI) and is co-funded under the European Regional Development Fund
under Grant Number 13/RC/2077.}\thanks{Quang-Doanh Vu, and Markku Juntti are with Centre for Wireless Communications,
University of Oulu, FI-90014, Finland. Email: \{quang.vu, markku.juntti\}@oulu.fi.}\thanks{Een-Kee Hong is with the School of Electronics and Information, Kyung
Hee University, South Korea. Email: ekhong@khu.ac.kr.}\thanks{Le-Nam Tran is with the Department of Electronic Engineering, Maynooth
University, Maynooth, Co. Kildare, Ireland. Email: lenam.tran@nuim.ie.}}
\maketitle
\begin{abstract}
As a wide class of resource management problems in wireless communications
are nonconvex and even NP-hard in many cases, finding globally optimal
solutions to these problems is of little practical interest. Towards
more pragmatic approaches, there is a rich literature on iterative
methods aiming at finding a solution satisfying necessary optimality
conditions to these problems. These approaches have been derived under
several similar mathematical frameworks such as inner approximation
algorithm, concave-convex procedure, majorization-minimization algorithm,
and successive convex approximation (SCA). However, a large portion
of existing algorithms arrive at a relatively generic program at each
iteration, which is less computationally efficient compared to a more
standard convex formulation. This paper proposes \emph{numerically
efficient }transformations and approximations for SCA-based methods
to deal with nonconvexity in wireless communications design. More
specifically, the central goal is to show that various nonconvex problems
in wireless communications can be iteratively solved by conic quadratic
optimization. We revisit various examples to demonstrate the advantages
of the proposed approximations. Theoretical complexity analysis and
numerical results show the superior efficiency in terms of computational
cost of our proposed solutions compared to the existing ones.
\end{abstract}


\section{Introduction}

The exponential increase in the number of portable devices which have
powerful multimedia capabilities (e.g.\ smartphones) has given rise
to tremendous demand on wireless data traffic. For example, Cisco\textquoteright s
projections of global mobile data traffic predicts smartphone will
reach three-quarters of mobile data traffic by 2019 and global mobile
data traffic will increase nearly tenfold between 2014 to 2019 \cite{ciscowhitepaper}.
In addition, more and more connected devices and the scarcity of bandwidth
make the coordination of multiuser interference highly important while
complicated. The evolution of wireless networks also posts other challenges
including resource costs, environmental impact, security, fairness
between subcribers, etc.\ \cite{samsung5g,docomo5g}.

Over the years, advanced optimization techniques have been widely
used and become vital tools for wireless communications design \cite{convexopt_WYu_2006,convexbeam2010,mimoframPanoma2003,Emil_multiobjective_5G}.
Representative examples include semidefinite relaxation (SDR) technique
\cite{SDrelaxation_Tlou2010}, dual decomposition and alternating
direction method of multipliers (ADMM) \cite{decomposition,DADMM},
robust optimization\cite{robustpaloma2006}, to name but a few. In
general, the first choice of solving a resource management problem
is to represent (or equivalently reformulate) it in a form of a convex
program, if possible. A good example in this regard is the problem
of minimizing the transmit power at the transmitter while the quality-of-service
(QoS) of each individual receiver is guaranteed \cite{SOCP_linear_precoding2006}.
This power minimization problem can be solved optimally by transforming
the original nonconvex problem into an equivalent second-order cone
program (SOCP). Unfortunately, such efficient convex reformulation
is impossible for many other design problems, e.g., weighted sum rate
maximization \cite{tomLuocomplexity2008,optimalWSRMISO}, energy efficiency
maximization for multiuser systems \cite{EE_oskari}, full-duplex
communications \cite{smallcell_Dan}, relay networks \cite{Jorswieck:relay2016},
etc. Generally, finding a globally optimal solution to these nonconvex
problems is difficult and, more importantly, not practically appealing.
Consequently, low-complexity suboptimal approaches are of particular
interest.

Among suboptimal solutions in the literature, SDR and successive convex
approximation (SCA) techniques are the two that have been extensively
used to tackle the nonconvexity in various wireless communication
problems \cite{SDrelaxation_Tlou2010,Nam:WSRMISO:2012,feasibleSAC,scutarisca2}.
Basically, instead of dealing with a design parameter, say $\mathbf{x}$,
the SDR method defines a positive semidefinite (PSD) matrix $\mathbf{X}\triangleq\mathbf{x}\mathbf{x}\herm$
and then lifts the design problem into the PSD domain. By omitting
the rank-1 constraint on $\mathbf{X}$, we can arrive at a semidefinite
program (SDP). For some special cases, a SDR-based solution can yield
exact optimal solutions, see \cite{SDrelaxation_Tlou2010} and references
therein. In other cases, randomization techniques are required to
produce high-performance feasible solutions. The major disadvantage
of the SDR is that the computational complexity of the resulting SDP
increases quickly with the problem size. Moreover, randomization techniques
might not always provide feasible points, and thus the SDR even fails
to obtain a feasible solution. The two issues were investigated in
\cite{Anh:Cogni:2012,feasibleSAC,Conic_quad}.

This paper is centered on a class of suboptimal solutions which are
based on SCA. In fact, SCA is a general term referred to similar algorithms
such as inner approximation algorithm \cite{Inner_Approximation1978},
concave-convex procedure \cite{concave:convexprod2003}, majorization-minimization
algorithm \cite{MajorM2003}, or difference of convex (DC) algorithm
\cite{DCbook}. Essentially, the idea of a SCA-based approach is to
\emph{safely and iteratively} approximate the nonconvex feasible set
(and/or nonconvex objective) of a nonconvex problem by a convex one
\cite{ABeck:SCAmath:2010,Inner_Approximation1978,concave:convexprod2003,scutariscaI}.
By a proper approximation, SCA is provably monotonically convergent
to a stationary solution to the original nonconvex problem. Compared
to SDR, SCA is more flexible and can be applied to a broader range
of applications. However, since an SCA-based solution requires to
solve a sequence of convex programs, the type of convex programs significantly
affects the computation time of the SCA-based method. In other words,
the overall complexity of SCA-based solutions strongly depends on
that of the convex program arrived at each iteration.

In general, a convex problem can be efficiently solved, i.e., in polynomial
time, by interior-point methods. However, the exponents of the polynomial
vary significantly according to the \emph{structure} of the convex
program \cite{lecture_on_modernCO}. For example, solving a linear
program is much more efficient than solving an SOCP in terms of both
complexity and stability \cite[Chap. 6]{lecture_on_modernCO}. In
the same way, SOCP's are much more computationally efficient than
SDP's. Thus one would consider an SOCP instead of an equivalent SDP,
if possible, and many examples are given in \cite{socp_Alizadeh,SOCApp_boyd1999}.

Motivated by the above discussions, we propose in this paper novel
transformations and approximations particularly useful for wireless
communications design where the problems of interest are nonconvex.
Different from many seminal papers above in which general algorithmic
frameworks for SCA are the main focus, we are interested in providing
formulations that allow for solving a wide variety of problems by
conic quadratic programming. The choice of conic quadratic programming
(CQP) is affected by the fact that linear programming (LP) is nearly
impossible, as far as beamforming techniques in multiple antennas
systems are concerned. Our contributions include the following:
\begin{itemize}
\item We first present a comprehensive review on the general framework of
SCA. A possible method of finding a feasible point to start the SCA
is also described. Then, we identify a class of common nonconvex constraints
in wireless communications design and propose transformations and
approximations to convert them to conic quadratic constraints.
\item In the second part of the paper, we demonstrate the efficiency and
flexibility of the proposed formulations in dealing with various resource
management problems in wireless communications. These include physical
layer secure transmission, relay communications, cognitive radio,
multicarrier management with different design criteria such as rate
maximization, transmit power minimization, rate fairness, and energy
efficiency fairness. Analytical and numerical results are provided
to show the superior performance of the proposed solutions, compared
to the existing ones.
\end{itemize}
The rest of the paper is organized as follows. Section \ref{sec:Preliminaries}
provides the preliminaries of the SCA. Section \ref{sec:Second-order-Cone-Programming}
presents the proposed low-complexity formulations. Section \ref{sec:SecureAF}
illustrates an example of using the proposed approach for the scenario
of secrecy relay transmissions. The second application about cognitive
radio is studied in Section \ref{sec:Application-II:-Cognitive}.
Section \ref{sec:Application-III:-MIMO} shows how to particularize
the proposed method to the case of MIMO relay communications. In Section
\ref{sec:Application-IV:-Power}, we consider the scenario of multiuser
multicarrier system with two problems: weighted sum rate maximization
and max-min energy efficiency fairness. Finally, Section \ref{sec:Conclusion}
concludes the paper.

\emph{Notation}: Standard notations are used in this paper. Bold lower
and upper case letters represent vectors and matrices, respectively;
$\left\Vert \cdot\right\Vert _{2}$ represents the $\ell_{2}$ norm;
$\left|\cdot\right|$ represents the absolute value; $\cdot^{\ast}$
represents the complex conjugate; $\mathbb{R}^{m\times n}$ and $\mathbb{C}^{m\times n}$,
represent the space of real and complex matrices of dimensions given
in superscript, respectively; $\mathbf{I}_{n}$ denotes the $n\times n$
identity matrix; $\mathcal{CN}(0,c\mathbf{I})$ denotes a complex
Gaussian random vector with zero mean and variance matrix $c\mathbf{I}$;
$\mathrm{\Re}(\cdot)$ and $\Im(\cdot)$ represents real and image
parts of the argument; $\mathbf{X}\herm$ and $\mathbf{X}\trans$
are Hermitian and normal transpose of $\mathbf{X}$, respectively;
$\myvec(\mathbf{X})$ is the vectorization operation that converts
the matrix $\mathbf{X}$ into a column vector; $\mathrm{trace}(\mathbf{X})$
is the trace of $\mathbf{X}$; $\mathbf{X}\otimes\mathbf{Y}$ denotes
Kronecker product. For ease of description, we also use \textquotedblleft MATLAB
notation\textquotedblright{} throughout the paper. Specifically, when
$\mathbf{X}_{1}$, ..., $\mathbf{X}_{k}$ are matrices with the same
number of rows, $[\mathbf{X}_{1},...,\mathbf{X}_{k}]$ denotes the
matrix with the same number of rows obtained by staking horizontally
$\mathbf{X}_{1}$, ..., and $\mathbf{X}_{k}$. When $\mathbf{X}_{1}$,
..., $\mathbf{X}_{k}$ are matrices with the same number of columns,$[\mathbf{X}_{1};...;\mathbf{X}_{k}]$
stands for the matrix with the same number of columns obtained by
staking vertically $\mathbf{X}_{1}$, ..., and $\mathbf{X}_{k}$.
Other notations are defined at their first appearance.

\section{Preliminaries \label{sec:Preliminaries}}

\subsection{Successive Convex Approximation \label{subsec:Successive-Convex-Approximation}}

We first briefly review some preliminaries of SCA which are central
to the discussions presented in the rest of the work. The interested
reader is referred to \cite{ABeck:SCAmath:2010,Inner_Approximation1978,scutariscaI}
for further details. Note that the SCA presented in this paper can
include concave-convex procedure \cite{concave:convexprod2003}, majorization-minimization
(MM) algorithm \cite{MajorM2003,MMpalomar}, or DC algorithm \cite{DCbook}
as special cases.\footnote{More precisely, when the feasible set of the problem being considered
is convex, MM algorithms are different from SCA in the sense that
the surrogate function in MM framework can be nonconvex. The interested
reader is referred to \cite{MMpalomar} for a sharp comparison between
SCA and MM algorithms. } Let us consider a general nonconvex optimization problem of the following
form \begin{subequations}\label{eq:SCAoverview}
\begin{align}
\underset{\mathbf{x}\in\mathbb{C}^{n}}{\mini} & \;f_{0}(\mathbf{x})\\
\st & \;f_{i}(\mathbf{x})\leq0,\;i=1,\ldots,L_{1}\label{eq:scca_nonconvex1}\\
 & \;\tilde{f}_{j}(\mathbf{x})\leq0,\;j=1,\ldots,L_{2}\label{eq:sca_nonconvex}
\end{align}
\end{subequations}  where $f_{i}(\mathbf{x}):\mathbb{C}^{n}\rightarrow\mathbb{R}$,
$i=0,\ldots,L_{1}$ and $\tilde{f}_{j}(\mathbf{x}):\mathbb{C}^{n}\rightarrow\mathbb{R}$,
$j=1,\ldots,L_{2}$ are assumed to be convex and noncovex functions,
respectively. We also assume that all the functions in \eqref{eq:SCAoverview}
are continuously differentiable. Note that \eqref{eq:SCAoverview}
is in the complex domain as it appears naturally in wireless communications
design (the case of the real domain or mixed real-complex domain will
be elaborated later on). In \eqref{eq:SCAoverview}, $f_{0}(\mathbf{x})$
represents the cost function to be optimized, and $f_{i}(\mathbf{x})\leq0,\;i=1,\ldots,L_{1}$,
$\tilde{f}_{j}(\mathbf{x})\leq0\;j=1,\ldots,L_{2}$ are the constraints
related to, e.g., quality of service (minimum rate requirement), radio
resources (transmit power, system bandwidth, backhaul), etc. At first,
the assumption of convexity on $f_{0}(\mathbf{x})$ seems to be strong,
but as we will see below, the cost functions in many design problems
are originally convex. Moreover, by proper transformations, e.g.,
using the epigraph form, we can bring the nonconvexity of the objective
into the feasible set.\footnote{It is not difficult to see that the epigraph form also has the same
KKT points as \eqref{eq:SCAoverview}. We refer the interested reader
to \cite{scutariscaI,scutarisca2} for a direct way to deal with problems
with a nonconvex objective. Herein the assumption on the convexity
of $f_{0}(\mathbf{x})$ is mainly to simplify the general description
of SCA.}

The difficulty of solving \eqref{eq:SCAoverview} is obviously due
to the nonconvex constraints in \eqref{eq:sca_nonconvex}. An SCA-based
approach is an iterative procedure which tries to seek a stationary
solution (i.e.\ a Karush-Kuhn-Tucker (KKT) point) of \eqref{eq:SCAoverview}
by sequentially approximating the nonconvex feasible set by \emph{inner}
convex ones. To do so, the nonconvex functions $\tilde{f}_{j}(\mathbf{x}),\;\forall j$,
are replaced by their convex upper bounds. We denote by $\mathbf{x}^{(\theta)}$
the solution obtained in the $\theta$th iteration of the iterative
process. At the iteration $\theta$, let $F_{j}(\mathbf{x};\mathbf{y}^{(\theta)}):\mathbb{C}^{n}\rightarrow\mathbb{R}$
denote a convex upper estimate of $\tilde{f}_{j}(\mathbf{x})$ which
is continuous on $(\mathbf{x};\mathbf{y}^{(\theta)})$, where $\mathbf{y}^{(\theta)}$
is a \emph{fixed} feasible point depending on the solution of the
problem in the $(\theta-1)$th iteration. That is $\mathbf{y}^{(\theta)}=g(\mathbf{x}^{(\theta-1)})$,
where $g(\mathbf{x}):\mathbb{C}^{n}\rightarrow\mathbb{C}^{m}$ is
a continuous function \cite{ABeck:SCAmath:2010}. Then the problem
considered in the $\theta$th iteration is given by\begin{subequations}\label{eq:SCAoverviewappx}
\begin{align}
\underset{\mathbf{x}\in\mathbb{C}^{n}}{\mini} & \;f_{0}(\mathbf{x})\\
\st & \;f_{i}(\mathbf{x})\leq0,\;i=1,\ldots,L_{1}\label{eq:scaconstr1}\\
 & \;F_{j}(\mathbf{x};\mathbf{y}^{(\theta)})\leq0,\;j=1,\ldots,L_{2}.\label{eq:scaconstr2}
\end{align}
\end{subequations} Note that problem \eqref{eq:SCAoverviewappx}
produces an upper bound of \eqref{eq:SCAoverview} due to the replacement
of $\tilde{f}_{j}(\mathbf{x})$ by $F_{j}(\mathbf{x},\mathbf{y}^{(\theta)})$.
To guarantee that the sequence of objective value $\{f_{0}(\mathbf{x}^{(\theta)})\}_{\theta}$
is nonincreasing and a limit point of the iterates $\{\mathbf{x}^{(\theta)}\}_{\theta}$,
if converge, is a stationary solution of \eqref{eq:SCAoverview},
the mapping function $\mathbf{y}=g(\mathbf{x})$ must satisfy the
following properties\begin{subequations}\label{eq:SCA:condition}
\begin{eqnarray}
\tilde{f}_{j}(\mathbf{x}') & \leq & F_{j}(\mathbf{x}';\mathbf{y})\label{eq:scacondition1}\\
\tilde{f}_{j}(\mathbf{x}) & = & F_{j}(\mathbf{x};\mathbf{y})\label{eq:scacondition2}\\
\nabla_{\mathbf{x}^{\ast}}\tilde{f}_{j}(\mathbf{x}) & = & \nabla_{\mathbf{x}^{\ast}}F_{j}(\mathbf{x};\mathbf{y})\label{eq:scacondition3}
\end{eqnarray}
\end{subequations} for all $j$, where $\nabla_{\mathbf{x}^{\ast}}f()$
denotes the gradient of $f$ with respect to the complex conjugate
of $\mathbf{x}$ (we for brevity refer to this gradient as the term
`conjugate gradient' in the rest of the paper). We note that the gradient
in \eqref{eq:scacondition3} can be replaced by the directional derivative
in general. In this regard, the definition of stationarity can be
made more specific (cf. \cite{Razaviyayn2013,Pang2017} for more details).
The main steps of the SCA procedure are outlined in Algorithm \ref{alg:generalSCA}.
\begin{algorithm}
\caption{Generic framework of an SCA-based approach}
\label{alg:generalSCA} \begin{algorithmic}[1]

\renewcommand{\algorithmicrequire}{\textbf{Initialization:}}

\REQUIRE Set $\theta=0$ and choose a feasible point $\mathbf{x}^{(0)}$.

\REPEAT

\STATE Solve \eqref{eq:SCAoverviewappx} to obtain optimal solution
$\mathbf{x}^{(\theta)}$.

\STATE Update: $\mathbf{y}^{(\theta+1)}\coloneqq g(\mathbf{x}^{(\theta)})$
and $\theta:=\theta+1$. \label{alg:generalSCA:update}

\UNTIL {Convergence.} \renewcommand{\algorithmicrequire}{\textbf{Output:}}

\REQUIRE $\mathbf{x}^{(\theta)}$.

\end{algorithmic}
\end{algorithm}

\subsubsection*{The real and mixed real-complex domain}

Before going forward we remark that for the case when $\mathbf{x}$
is a real valued vector, the conjugate gradient of $f$ in \eqref{eq:scacondition3}
is to be replaced by the normal gradient. If $\mathbf{x}=[\mathbf{x}_{1};\mathbf{x}_{2}]$
where $\mathbf{x}_{1}\in\mathbb{C}^{n_{1}}$ and $\mathbf{x}_{2}\in\mathbb{R}^{n_{2}}$,
$n_{1}+n_{2}=n$, then $\nabla_{\mathbf{x}^{\ast}}f(\mathbf{x})$
is defined as $\nabla_{\mathbf{x}^{\ast}}f(\mathbf{x})\triangleq[\nabla_{\mathbf{x}_{1}^{\ast}}f(\mathbf{x});\nabla_{\mathbf{x}_{2}}f(\mathbf{x})]$.
Furthermore, we define the inner product of two vectors $\mathbf{x}$
and $\mathbf{y}$, denoted by $\left\langle \mathbf{x},\mathbf{y}\right\rangle $,
for several cases. If $\mathbf{x}$ and $\mathbf{y}$ are real valued,
then $\left\langle \mathbf{x},\mathbf{y}\right\rangle \triangleq\mathbf{x}\trans\mathbf{y}$.
On the other hand, when $\mathbf{x}$ and $\mathbf{y}$ are complex
valued, then $\left\langle \mathbf{x},\mathbf{y}\right\rangle \triangleq2\Re(\mathbf{x}\herm\mathbf{y})$.

\subsection{Initial Feasible Points for Algorithm \ref{alg:generalSCA}: A Relaxed
Algorithm}

We can see that Algorithm \ref{alg:generalSCA} requires an initial
feasible point $\mathbf{x}^{(0)}$ to start the iterative procedure
of the SCA. In some problems such as those where all $f_{i}$ and
$\tilde{f}_{j}$ are homogeneous, generating a feasible point can
be done easily via \emph{scaling/rescaling} operation. More specifically,
we can randomly generate a vector $\mathbf{x}^{(0)}\in\mathbb{R}^{n}$
and then multiply the vector by a proper scalar such that all the
constraints are satisfied. An example will be shown in Section \ref{sec:SecureAF}.
However, such simple manipulations cannot apply to more sophisticated
cases. To overcome this issue, we present a relaxed version of Algorithm
\ref{alg:generalSCA} which was used in \cite{Dinh:SCP:2011,feasibleSAC,doanh_spl_2015,Lipp:CCP:2016,advanceADC}.
Consider a relaxed problem of \eqref{eq:SCAoverviewappx} given by
\begin{subequations}\label{eq:SCAoverviewappx-1}
\begin{align}
\underset{\mathbf{x}\in\mathbb{C}^{n},q\geq0}{\mini} & \;f_{0}(\mathbf{x})+\lambda q\label{eq:scaapr_reg:obj}\\
\st & \;f_{i}(\mathbf{x})\leq q,\;i=1,\ldots,L_{1}\label{eq:scaapr_reg}\\
 & \;F_{j}(\mathbf{x},\mathbf{y}^{(\theta+1)})\leq q,\;j=1,\ldots,L_{2}\label{eq:scaapr_reg1}
\end{align}
\end{subequations}where $q$ is the newly introduced slack variables
and $\lambda$ is a positive parameter. The purpose of introducing
$q$ is to make \eqref{eq:SCAoverviewappx-1} feasible for any $\mathbf{x}^{(0)}\in\mathbb{C}^{n}$.
Indeed, given some $\mathbf{x}^{(0)}$, we can always find $q$ with
sufficiently large elements satisfying \eqref{eq:scaapr_reg} and
\eqref{eq:scaapr_reg1}. In addition, let $(\mathbf{x},q)$ be a feasible
point of \eqref{eq:SCAoverviewappx-1}. Then $\mathbf{x}$ is also
feasible for \eqref{eq:SCAoverview} if $q=0$. That is to say, successively
solving \eqref{eq:SCAoverviewappx-1} may produce a feasible solution
to \eqref{eq:SCAoverview} because $q$ is encouraged to be zero due
to the minimization of the objective in \eqref{eq:SCAoverviewappx-1}.
Moreover, if $f_{0}(\mathbf{x})$ is strongly convex and the Mangasaran-Fromovitz
constraint qualification is satisfied at any $\mathbf{x}'\in\mathbb{C}^{n}$
such that $\max(\{f_{i}(\mathbf{x}')\}_{i},\{\tilde{f}_{j}(\mathbf{x}')\})\geq0$,
then it guarantees that $q=0$ after a finite number of iterations
if $\lambda$ is larger than a finite lower bound which can be determined
based on the Lagrangian multipliers of \eqref{eq:SCAoverviewappx}
\cite{advanceADC}. In general this lower bound is difficult to find
analytically and the choice of $\lambda$ can be heuristic in practice.
For example, $\lambda$ can be set to be sufficiently large \cite{Dinh:SCP:2011},
or it can be increased after each iteration \cite{Dinh:SCP:2011,Lipp:CCP:2016}.
Note that these results still hold if we replace the variable $q$
by $q_{i}$ in \eqref{eq:scaapr_reg} for $i=1,\ldots,L_{1}$, by
$\tilde{q}_{j}$ in \eqref{eq:scaapr_reg1} for $j=1,\ldots,L_{2}$,
and the term $\lambda q$ by $\lambda\bigl(\sum_{i=1}^{L_{1}}q_{i}+\sum_{j=1}^{L_{2}}\tilde{q}_{j}\bigr)$
in \eqref{eq:scaapr_reg:obj} \cite{advanceADC}. We summarize this
relaxed procedure in Algorithm \ref{alg:generate-initial-SCA}.

\begin{algorithm}
\caption{A general initialization for Algorithm \ref{alg:generalSCA}}
\label{alg:generate-initial-SCA} \begin{algorithmic}[1]

\renewcommand{\algorithmicrequire}{\textbf{Initialization:}}

\REQUIRE Randomly generate $\mathbf{x}^{(0)}\in\mathbb{C}^{n}$.

\REPEAT

\STATE Solve \eqref{eq:SCAoverviewappx-1} to obtain optimal solution
$(\mathbf{x}^{(\theta)},q^{(\theta)})$.

\STATE Update $\mathbf{y}^{(\theta+1)}\coloneqq g(\mathbf{x}^{(\theta)})$
and $\theta:=\theta+1$.

\UNTIL {$q^{(\theta)}=0$} \renewcommand{\algorithmicrequire}{\textbf{Output:}}

\REQUIRE$\mathbf{x}^{(\theta)}$.

\end{algorithmic}
\end{algorithm}

\subsection{Convergence Results }

There are several convergence results of the SCA \cite{Inner_Approximation1978,concave:convexprod2003,MajorM2003,ABeck:SCAmath:2010,Dinh:SCP:2011,scutari_dc_2014,DCbook,scutariscaI,thesissca},
among them \cite{Inner_Approximation1978,concave:convexprod2003,ABeck:SCAmath:2010,Dinh:SCP:2011,DCbook,scutariscaI,thesissca}
consider nonconvex constraints. We briefly summarize these results
herein for the sake of completeness. First, due to the use of upper
bounds in \eqref{eq:scaconstr2}, $\mathbf{x}^{(\theta+1)}$ is feasible
to the convexified subproblem \eqref{eq:SCAoverviewappx} at iteration
$k$. Thus $f_{0}(\mathbf{x}^{(\theta)})\geq f_{0}(\mathbf{x}^{(\theta+1)})$
and the sequence $\{f_{0}(\mathbf{x}^{(\theta)})\}_{\theta=1}^{\infty}$
is nonincreasing, possibly to negative infinity. If $f_{0}(\mathbf{x})$
is coercive on the feasible set or the feasible set is bounded, then
$\{f_{0}(\mathbf{x}^{(\theta)})\}_{\theta}$ converges to a finite
value. However, this does not necessarily imply the convergence of
the iterates $\{\mathbf{x}^{(\theta)}\}_{\theta=0}^{\infty}$ to a
stationary point or a local minimum in general.\footnote{If $\{\mathbf{x}^{(\theta)}\}_{\theta=0}^{\infty}$ converges, the
limit point is a stationary solution of \eqref{eq:SCAoverview} \cite[Th. 1]{ABeck:SCAmath:2010,thesissca}.}

To establish the convergence of the iterates, the key point is to
ensure that a strict descent can be obtained after each iteration,
i.e., $f_{0}(\mathbf{x}^{(\theta)})>f_{0}(\mathbf{x}^{(\theta+1)})$
for all $\theta$, unless $\mathbf{x}^{(\theta)}=\mathbf{x}^{(\theta+1)}$.
If $f_{0}(\mathbf{x})$ is strongly convex and the feasible set is
convex as assumed in \cite{ABeck:SCAmath:2010}, the strict descent
property and convergence of the iterates to a stationary point are
guaranteed. When strict convexity does not hold for $f_{0}(\mathbf{x})$,
the iterative process might not make the objective sequence strictly
decreasing. To overcome this issue, we can add a \emph{proximal term}
as done in \cite{Dinh:SCP:2011}. In particular, instead of \eqref{eq:SCAoverviewappx},
we consider a regularized problem of \eqref{eq:SCAoverviewappx} given
by \begin{subequations}\label{eq:SCAproximal}
\begin{align}
\underset{\mathbf{x}\in\mathbb{C}^{n}}{\mini} & \;f_{0}(\mathbf{x})+\alpha||\mathbf{x}-\mathbf{x}^{(\theta-1)}||_{2}^{2}\\
\st & \;\eqref{eq:scaconstr1},\eqref{eq:scaconstr2}\label{eq:scaapr_reg-1}
\end{align}
\end{subequations} where $\alpha>0$ is a regularization parameter.
By adding proximal term $\alpha||\mathbf{x}-\mathbf{x}^{(\theta)}||_{2}^{2}$,
$\{f_{0}(\mathbf{x}^{(\theta)})\}_{\theta}$ is strictly decreasing
and the sequence $\{\mathbf{x}^{(\theta)}\}_{\theta}$ converges to
a stationary point due to the following result $f_{0}(\mathbf{x}^{(\theta)})-f_{0}(\mathbf{x}^{(\theta+1)})\geq\alpha||\mathbf{x}^{(\theta+1)}-\mathbf{x}^{(\theta)}||_{2}^{2}$
\cite{Dinh:SCP:2011}. In fact, the proximal term makes each subproblem
of SCA-based methods strongly convex, which is the main idea behind
the work of \cite{scutariscaI}. Note that parameter $\alpha$ should
not be large, otherwise, the algorithm converges slowly. In practice,
we should only consider \eqref{eq:SCAproximal} if a strict descent
is not achieved at the current iteration \cite{Dinh:SCP:2011}. Stronger
convergence results for problems where the objective and constraints
are nonconvex were reported in \cite{scutariscaI}. Particularly,
if the approximation function of the nonconvex cost function is strongly
convex (but not necessarily a tight global upper bound) and some other
mild conditions are satisfied, the SCA procedure is guaranteed to
converge to a stationary point with appropriate rules of updating
operation points. Moreover, possibilities for distributed solutions
under the framework of SCA were also discussed in \cite[Section IV]{scutariscaI}.

\subsection{Desired Properties of Approximation Functions\label{subsec:Choice-of-Approximation}}

The main point of an SCA-based method is to find a convex upper approximation
for a nonconvex function $\tilde{f}(\mathbf{x})$ that satisfies the
three conditions in \eqref{eq:SCA:condition}. There are in fact several
ways to do this. If $\tilde{f}(\mathbf{x})$ is concave, constraint
$\tilde{f}(\mathbf{x})\leq0$ is called a reverse one and a convex
upper bound of $\tilde{f}(\mathbf{x})$ can be easily found from the
first order Taylor series. In many cases, $\tilde{f}(\mathbf{x})$
is a DC function, i.e., $\tilde{f}(\mathbf{x})=h(\mathbf{x})-g(\mathbf{x})$
where both $h(\mathbf{x})$ and $g(\mathbf{x})$ are convex. In such
a case, a convex upper bound can be given by $F(\mathbf{x};\mathbf{y})=h(\mathbf{x})-g(\mathbf{y})-\left\langle \nabla_{\mathbf{x}^{\ast}}g(\mathbf{y}),\mathbf{x}-\mathbf{y}\right\rangle $,
which is usually done in the context of the convex-concave procedure
\cite{concave:convexprod2003,Lipp:CCP:2016}. Note that there are
infinitely many DC decompositions for $\tilde{f}(\mathbf{x})$, e.g,
by adding a quadratic term to both $h(\mathbf{x})$ and $g(\mathbf{x})$.
In particular, if the gradient of $\tilde{f}(\mathbf{x})$ is $L$-Lipschitz
continuous, a surrogate function is given by $F(\mathbf{x};\mathbf{y})=\tilde{f}(\mathbf{y})+\left\langle \nabla_{\mathbf{x}}\tilde{f}(\mathbf{y}),\mathbf{x}-\mathbf{y}\right\rangle +\frac{L}{2}||\mathbf{x}-\mathbf{y}||_{2}^{2}$
which is a convex quadratic function. In some problems, the range
of $\mathbf{x}$ may be useful to find a convex upper bound of $\tilde{f}(\mathbf{x})$.

From the above discussion, one would be interested in finding the
best approximation function for a given nonconvex function. However,
the solution to this problem is not unique as it is problem specific.
In general, a good approximation function will provide at least two
features: \emph{tightness} and \emph{numerical tractability}. The
tightness property is obvious as we want the approximate convex set
and the original nonconvex one are as close as possible. This has
a huge impact on the convergence rate of the iterative procedure.
The numerical tractability properties means that the chosen approximation
should yield a convex program that can be solved very efficiently,
e.g., by analytical solution. In case an analytical solution is impossible
to find, we may prefer to seek an approximation such that the convex
subproblem in (2) belongs to a class of convex programs for which
numerical methods are known to be more efficient. For example, an
SOCP is much easier to solve than a generic convex problem consisting
of a mix of SOC and exponential cone constraints. To see our arguments,
let us consider an exemplary problem $\min\{\tilde{f}(\mathbf{x})|\mathbf{A}\mathbf{x}\leq\mathbf{b},\mathbf{x}\in\mathbb{R}_{+}^{n}\}$
where $\tilde{f}(\mathbf{x})=\sum_{i=1}^{n}\log(\mathbf{u}_{i}^{T}\mathbf{x})-\sum_{i=1}^{n}\log(\mathbf{v}_{i}^{T}\mathbf{x})$
and $\mathbf{A}$, $\mathbf{b}$, $\mathbf{u}$, $\mathbf{v}$ are
the problem data. An easy and straightforward approximation of $f(\mathbf{x})$
would be by linearizing the term $\sum_{i=1}^{n}\log(\mathbf{v}_{i}^{T}\mathbf{x})$,
which was actually considered in \cite{Kha:dcprog:2012,maxminEE:SCA:TVT:2015}.
However, this results in a generic nonlinear program (NLP). By exploiting
the problem structure, we may find a positive $c$ such that $c||\mathbf{x}||_{2}^{2}-\tilde{f}(\mathbf{x})$
is convex, and thus $\tilde{f}(\mathbf{x})=c||\mathbf{x}||_{2}^{2}-(c||\mathbf{x}||_{2}^{2}-\tilde{f}(\mathbf{x}))$
is in a DC form. As a result, a convex quadratic program (QP) is obtained
if we linearize the term $c||\mathbf{x}||_{2}^{2}-\tilde{f}(\mathbf{x})$.
Obviously, a QP is easier to solve than a generic NLP in terms of
solution efficiency. An interesting example for this is provided in
Section \ref{sec:Application-IV:-Power}.

\section{Proposed Conic Quadratic Approximate Formulations \label{sec:Second-order-Cone-Programming}}

As discussed above, there are infinitely many approximations for a
given nonconvex function or a nonconvex constraint, which have crucial
impact on the convergence speed, numerical efficiency, etc.\ of SCA-based
algorithms. From the standpoint of numerical optimization methods,
it is probably best to arrive at a linear program for each subproblem
in a SCA-based method. Unfortunately, this is hard to achieve in many
wireless communications related problems, especially in view of beamforming
techniques for multiantenna systems. As a result, conic quadratic
optimization is a good choice due to its broad modeling capabilities
and computational stability and efficiency. In this section, we will
present some nonconvex constraints widely seen in wireless communications
design problems and introduce novel convex approximations to deal
with their nonconvexity.

The nonconvex constraints considered in this paper are given in a
general form as
\begin{equation}
l\leq\frac{h_{1}(\mathbf{x})}{h_{2}(\mathbf{x})}\leq u\label{eq:doubleside}
\end{equation}
where $h_{1}(\mathbf{x})$ and $h_{2}(\mathbf{x})$ are affine or
convex quadratic functions. We assume that $l>0$, $h_{1}(\mathbf{x})>0$,
and $h_{2}(\mathbf{x})>0$, which hold in numerous practical problems
in wireless communications. The upper and lower limits $u\in\mathbb{R}$
and $l\in\mathbb{R}$ may be constants or optimization variables,
depending on the specific problem. The cases where $u$ and/or $l$
are optimization variables mostly result from considering the epigraph
form of the original design problem.

\subsection{Case 1: $h_{1}(\mathbf{x})$ and $h_{2}(\mathbf{x})$ Are Affine}

We note that, when $h_{1}(\mathbf{x})$ and $h_{2}(\mathbf{x})$ are
affine, if $u$ or $l$ is a constant, the associated constraint becomes
a linear one, thus approximation is not needed. Here we are interested
in the case where both $u$ and $l$ are optimization variables. In
this case, $\mathbf{x}$ is a real-valued vector. This class of constraints
usually occurs in power control problems \cite{power_controlGP}.
To handle such nonconvex constraints, \cite{power_controlGP} applied
SCA so that the nonconvex problem is approximated as geometric programming
(GP). We now show that this constraint can be approximated as a conic
quadratic formulation.

Let us consider the constraint $\frac{h_{1}(\mathbf{x})}{h_{2}(\mathbf{x})}\geq l>0$
first, which is equivalent to $f(\tilde{\mathbf{x}})=lh_{2}(\mathbf{x})-h_{1}(\mathbf{x})\leq0$
with $\tilde{\mathbf{x}}=[\mathbf{x};l]$. A convex upper bound can
be found
\begin{equation}
f(\tilde{\mathbf{x}})\leq F(\tilde{\mathbf{x}};y)=\frac{y}{2}l^{2}+\frac{1}{2y}h_{2}^{2}(\mathbf{x})-h_{1}(\mathbf{x})\tag{App1}\label{eq:App1}
\end{equation}
and the SCA parameter update in Step \ref{alg:generalSCA:update}
of Algorithm \ref{alg:generalSCA} is $y=\frac{h_{2}(\mathbf{x})}{l}$.
Note that $F(\tilde{\mathbf{x}};y)$ is a convex quadratic function
for a given $y$ and a generalization of a result in \cite{ABeck:SCAmath:2010}.
The gradient of the upper bound function at some $\tilde{\mathbf{x}}'=[\mathbf{x}';l']$
is given as $\nabla_{\tilde{\mathbf{x}}}F(\tilde{\mathbf{x}}';y)=[\frac{h_{2}(\mathbf{x}')}{y}\nabla_{\mathbf{x}}h_{2}(\mathbf{x}')-\nabla_{\mathbf{x}}h_{1}(\mathbf{x}');yl']$
which reduces to $[l'\nabla_{\mathbf{x}}h_{2}(\mathbf{x}')-\nabla_{\mathbf{x}}h_{1}(\mathbf{x}');h_{2}(\mathbf{x}')]=\nabla_{\tilde{\mathbf{x}}}f(\tilde{\mathbf{x}}')$
when substituting $y$ by $\frac{h_{2}(\mathbf{x}')}{l'}$. That is,
the condition \eqref{eq:scacondition3} is satisfied by the bound
in \eqref{eq:App1}.

In another way, we can rewrite $f(\tilde{\mathbf{x}})$ in a DC form,
i.e.\ $f(\tilde{\mathbf{x}})=\frac{1}{4}\left(l+h_{2}(\mathbf{x})\right)^{2}-\underset{\bar{h}(\tilde{\mathbf{x}})}{\underbrace{{\textstyle \frac{1}{4}}\left(l-h_{2}(\mathbf{x})\right)^{2}}}-h_{1}(\mathbf{x})$,
and the convex upper estimate can be found as
\begin{equation}
F(\tilde{\mathbf{x}};\mathbf{y})=\frac{1}{4}\left(l+h_{2}(\mathbf{x})\right)^{2}-\bar{h}(\mathbf{y})-\left\langle \nabla_{\tilde{\mathbf{x}}}\bar{h}(\mathbf{y}),\tilde{\mathbf{x}}-\mathbf{y}\right\rangle -h_{1}(\mathbf{x})\tag{App2}\label{eq:App2}
\end{equation}
In this case the SCA parameter update is simply $\mathbf{y}=\tilde{\mathbf{x}}.$

The constraint $f(\tilde{\mathbf{x}})\triangleq\frac{h_{1}(\mathbf{x})}{h_{2}(\mathbf{x})}-u\leq0$,
$\tilde{\mathbf{x}}=[\mathbf{x};u]$, can be handled similarly. Specifically,
in light of \eqref{eq:App1}, a convex approximation of $f(x)$ can
be found as
\begin{equation}
f(\tilde{\mathbf{x}})\leq F(\tilde{\mathbf{x}};y)=\frac{y}{2}h_{1}^{2}(\mathbf{x})+\frac{1}{2y}h_{2}^{-2}(\mathbf{x})-u\tag{App3}\label{eq:App3}
\end{equation}
where the SCA parameter update is $y=\frac{1}{h_{1}(\mathbf{x})h_{2}(\mathbf{x})}$.
Following the same steps as for \eqref{eq:App1}, it is straightforward
to check that the bound in \eqref{eq:App3} satisfies the condition
\eqref{eq:scacondition3}. Note that $F(\tilde{\mathbf{x}};\mathbf{y})$
is not a quadratic function but the constraint $F(\tilde{\mathbf{x}};\mathbf{y})\leq0$
can be easily expressed by the following two SOC ones
\begin{eqnarray}
\frac{y}{2}h_{1}^{2}(\mathbf{x})+\frac{1}{2y}z^{2}-u & \leq & 0\\
1 & \leq & h_{2}(\mathbf{x})z\label{eq:rotatedcone}
\end{eqnarray}
Note that \eqref{eq:rotatedcone} is a rotated SOC constraint, since
$h_{2}(\mathbf{x})$ is affine. Alternatively, we can simply use the
approximations similar to \eqref{eq:App2}.

\subsection{Case 2: $h_{1}(\mathbf{x})$ and $h_{2}(\mathbf{x})$ Are Convex
Quadratic Functions}

We now turn our attention to the case where $h_{1}(\mathbf{x})$ and
$h_{2}(\mathbf{x})$ are convex quadratic functions. In wireless communications
this form of constraint occurs in the problems related to precoder
designs and $\mathbf{x}$ is a vector of complex variables. Let us
consider the constraint $\frac{h_{1}(\mathbf{x})}{h_{2}(\mathbf{x})}\geq l$
first, which is equivalent to $f(\tilde{\mathbf{x}})\triangleq h_{2}(\mathbf{x})-\frac{h_{1}(\mathbf{x})}{l}\leq0$,
$\tilde{\mathbf{x}}=[\mathbf{x};l]$. Note that the term $\tilde{h}(\tilde{\mathbf{x}})\triangleq\frac{h_{1}(\mathbf{x})}{l}$
is convex with respect to $\tilde{\mathbf{x}}$, and thus a convex
approximation of $f(\tilde{\mathbf{x}})$ is given by
\begin{equation}
f(\tilde{\mathbf{x}})\leq F(\tilde{\mathbf{x}};\mathbf{y})=h_{2}(\mathbf{x})-\tilde{h}(\tilde{\mathbf{x}})-\left\langle \nabla_{\tilde{\mathbf{x}}^{\ast}}\tilde{h}(\mathbf{y}),\tilde{\mathbf{x}}-\mathbf{y}\right\rangle \tag{App4}\label{eq:App4}
\end{equation}
where the SCA parameter update is taken as $\mathbf{y}=\tilde{\mathbf{x}}.$
Another approximation can be found by introducing a slack variable,
i.e.\ we have
\begin{align}
\frac{h_{1}(\mathbf{x})}{h_{2}(\mathbf{x})}\geq l\Leftrightarrow\left\{ \begin{array}{l}
lz-h_{1}(\mathbf{x})\leq0\\
h_{2}(\mathbf{x})\leq z
\end{array}\right.
\end{align}
The first constraint in the equivalent formulation can be rewritten
as $f(\hat{\mathbf{x}})\triangleq(l+z)^{2}-\underset{\bar{h}(\hat{\mathbf{x}})}{\underbrace{((l-z)^{2}+4h_{1}(\mathbf{x}))}}\leq0$
where $\hat{\mathbf{x}}=[\mathbf{x};l;z]$, and $\bar{h}(\hat{\mathbf{x}})$
is a quadratic function with respect to $\hat{\mathbf{x}}$. Then
the approximation can be given by
\begin{equation}
f(\hat{\mathbf{x}})\leq F(\hat{\mathbf{x}};\mathbf{y})=(l+z)^{2}-\bar{h}(\mathbf{y})-\left\langle \nabla_{\hat{\mathbf{x}}^{\ast}}\bar{h}(\mathbf{y}),\hat{\mathbf{x}}-\mathbf{y}\right\rangle \tag{App5}\label{eq:App5}
\end{equation}
The SCA parameter is updated as $\mathbf{y}=\hat{\mathbf{x}}$.

For the constraint $\frac{h_{1}(\mathbf{x})}{h_{2}(\mathbf{x})}\leq u$
we can equivalently write it as $\frac{h_{1}(\mathbf{x})}{u}-h_{2}(\mathbf{x})\leq0$,
and then a convex upper estimator is simply given by
\begin{equation}
F([\mathbf{x},u];\mathbf{y})=\frac{h_{1}(\mathbf{x})}{u}-h_{2}(\mathbf{y})-\left\langle \nabla_{\mathbf{x}^{\ast}}h_{2}(\mathbf{y}),\mathbf{x}-\mathbf{y}\right\rangle \tag{App6}\label{eq:app6}
\end{equation}
with the SCA parameter update is $\mathbf{y}=\mathbf{x}.$ This constraints
can also be approximated using the same approach as that in \eqref{eq:App5}.

To conclude this section we now show the conjugate gradient appearing
in \eqref{eq:App4}, \eqref{eq:App5} and \eqref{eq:app6}. Let us
consider the conjugate gradient of $\tilde{h}(\tilde{\mathbf{x}})=\frac{h_{1}(\mathbf{x})}{l}$.
As mentioned earlier, we can write $\nabla_{\tilde{\mathbf{x}}^{\ast}}\tilde{h}(\mathbf{y})=\left[\frac{\nabla_{\mathbf{x}^{\ast}}h_{1}(\mathbf{y}_{x})}{y_{l}};-\frac{h_{1}(\mathbf{y}_{x})}{y_{l}^{2}}\right]$,
where $\mathbf{y_{\mathbf{x}}}$ and $y_{l}$ are the elements of
$\mathbf{y}$ corresponding to $\mathbf{x}$ and $l$ respectively.
That means, it requires the conjugate gradient of a quadratic function.
Suppose $h_{1}(\mathbf{x})=\mathbf{x}\herm\mathbf{A}\mathbf{x}+2\Re(\mathbf{b}^{H}\mathbf{x})+c$
where $\mathbf{A}$ is a PSD matrix, $\mathbf{b}\in\mathbb{C}^{n}$,
and $c\in R$. Then $\nabla_{\mathbf{x}^{\ast}}h_{1}(\mathbf{y})=\mathbf{A}\mathbf{y}+\mathbf{b}$.
The conjugate gradients in \eqref{eq:App5} and \eqref{eq:app6} follow
immediately.

Regarding the use of the above proposed algorithms we have the following
remarks
\begin{itemize}
\item When applying to a specific problem, an approximation function may
be better than another. Thus, we can consider all applicable approximations
to choose the best one for on-line design.
\item Many problems in wireless communications may not naturally express
the design constraints in the forms written in this paper, for which
cases equivalent transformations are required. In doing so, the number
of newly introduced variables should be kept minimal. This issue will
be further elaborated by an example in Section \ref{subsec:WSRmaxmulticarr}.
\end{itemize}
In the following sections, we apply the above approximations to address
four specific problems, which are chosen to cover a wide range of
scenarios in wireless communications. However, we note that the proposed
approximations also find applications in other contexts not considered
herein. For the numerical experiments to follow, we use the modeling
language YALMIP \cite{YALMIP} with MOSEK \cite{Mosek} being the
inner solver for SOCP, SDP, and GP. The proposed iterative method
stops when the increase (or decrease) of the last 5 consecutive iterations
is less than $10^{-3}$. The average run time reported in all figures
takes into account the total number of iterations for the iterative
algorithm to converge.

\section{Application I: Secure Beamforming Designs for Amplify-and-Forward
Relay Networks \label{sec:SecureAF}}

In the first application, we revisit the problem of secure beamforming
for amplify-and-forward (AF) relay networks which was studied in\cite{Yang:SecureAF:2013}.

\subsection{System Model and Problem Formulation}

\begin{figure}
\centering{}\includegraphics[width=0.87\columnwidth]{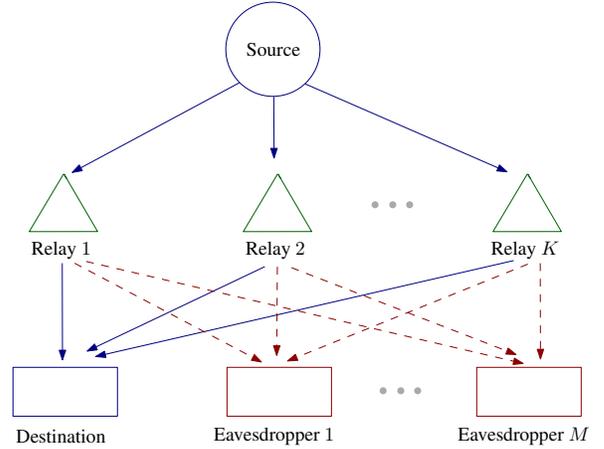}\caption{A diagram of secure AF relay networks with one source, one destination,
$K$ relays, and $M$ eavesdroppers.}
 \label{fig:AFsecuremodel}
\end{figure}

In the considered scenario, a source sends data to a destination through
the assistance of $K$ relays that operate in the AF mode. In addition,
there are $M$ eavesdroppers who want to intercept the information
intended for the destination. It is assumed that there is no direct
link between the source and the destination. The system model is illustrated
in Fig. \ref{fig:AFsecuremodel}. We adopt the notations used in \cite{Yang:SecureAF:2013}
for ease of discussion. Specifically, the channel between the source
and relay $k$, and that between relay $k$ and the destination are
denoted by $f_{k}$ and $g_{k}$, respectively. The channel between
relay $k$ and eavesdropper $m$ is denoted by $h_{km}$. Let $w_{k}$
be the complex weight used at relay $k$. For notational convenience,
the following vectors are defined: $\mathbf{g}\triangleq[g_{1},g_{2},\ldots,g_{K}]\herm\in\mathbb{C}^{K\times1}$,
$\mathbf{h}_{m}\triangleq[h_{1m},h_{2m},\ldots,h_{Km}]\herm\in\mathbb{C}^{K\times1}$,
$\mathbf{f}\triangleq[f_{1},f_{2},\ldots,f_{K}]\trans\in\mathbb{C}^{K\times1}$,
and $\mathbf{w}\triangleq[w_{1},w_{2},\ldots,w_{K}]\trans\in\mathbb{C}^{K\times1}$.
Let $\mathbf{n}_{\mathrm{r}}\sim\mathcal{CN}(0,\sigma^{2}\mathbf{I}_{K})$
be the noise vector at the relays. With the above notations, the signals
received at the destination and eavesdropper $m$ are
\begin{eqnarray}
y_{\mathrm{d}} & = & \sqrt{P_{\mathrm{s}}}\mathbf{g}\herm\D(\mathbf{f})\mathbf{w}s+\mathbf{n}_{\mathrm{r}}\trans\D\herm(\mathbf{g})\mathbf{w}+n_{\mathrm{d}}
\end{eqnarray}
and
\begin{equation}
y_{\mathrm{e},m}=\sqrt{P_{\mathrm{s}}}\mathbf{h}_{m}\herm\D(\mathbf{f})\mathbf{w}s+\mathbf{n}_{\mathrm{r}}\trans\D\herm(\mathbf{h}_{m})\mathbf{w}+n_{\mathrm{e},m}
\end{equation}
respectively, where $P_{\mathrm{s}}$ is the transmit power at the
source, $\D(\mathbf{x})$ denotes the diagonal matrix with elements
$\mathbf{x}$, $n_{\mathrm{d}}\sim\mathcal{CN}(0,\sigma^{2})$ and
$n_{\mathrm{e},m}\sim\mathcal{CN}(0,\sigma^{2})$ are the noise at
the destination and eavesdropper $m$, respectively. Then, the SINRs
at the destination and eavesdropper $m$ are given by
\begin{equation}
\gamma_{\mathrm{d}}=\frac{\mathbf{w}\herm\mathbf{A}\mathbf{w}}{\mathbf{w}\herm\mathbf{G}\mathbf{w}+1},\;\textrm{and}\ \gamma_{\mathrm{e},m}=\frac{\mathbf{w}\herm\mathbf{B}_{m}\mathbf{w}}{\mathbf{w}\herm\mathbf{H}_{m}\mathbf{w}+1}
\end{equation}
respectively, where $\mathbf{A}\triangleq(P_{\mathrm{s}}/\sigma^{2})\D\herm(\mathbf{f})\mathbf{gg}\herm\D(\mathbf{f})$,
$\mathbf{G}\triangleq\D(\mathbf{g})\D\herm(\mathbf{g})$, $\mathbf{B}_{m}\triangleq(P_{\mathrm{s}}/\sigma^{2})\D\herm(\mathbf{f})\mathbf{h}_{m}\mathbf{h}_{m}\herm\D(\mathbf{f})$,
and $\mathbf{H}_{m}\triangleq\D(\mathbf{h}_{m})\D\herm(\mathbf{h}_{m})$.
Now the problem of maximizing the secrecy rate reads \begin{subequations}\label{eq:secrecyrate}
\begin{align}
\underset{\mathbf{w}}{\maxi} & \;\underset{1\leq m\leq M}{\min}\left(\log(1+\gamma_{\mathrm{d}})-\log(1+\gamma_{\mathrm{e},m})\right)\\
\st & \;\mathbf{w}\herm\mathbf{C}\mathbf{w}\leq P_{\mathrm{tot}},\;\mathbf{w}\herm\D(\mathbf{e}_{k})\mathbf{C}\mathbf{w}\leq P_{k},\,\forall k\label{eq:powerconstrsecureaf}
\end{align}
\end{subequations} where $\mathbf{C}\triangleq P_{\mathrm{s}}\D\herm(\mathbf{f})\D(\mathbf{f})+\sigma^{2}\mathbf{I}_{K}$,
$\mathbf{e}_{k}\triangleq[\underset{k-1}{\underbrace{0,\ldots,0}},1,\underset{K-k}{\underbrace{0,\ldots,0}}]$,
$P_{\mathrm{tot}}$ is the maximum total transmit power for all the
relays, and $P_{k}$ is the maximum transmit power for relay $k$.

To solve \eqref{eq:secrecyrate}, \cite{Yang:SecureAF:2013} introduced
the PSD matrix $\mathbf{W}\triangleq\mathbf{w}\mathbf{w}\herm$ and
arrived at a relaxation of this problem where the rank-1 constraint
on $\mathbf{W}$ was dropped for tractability. Then the relaxed program
was solved by a method involving two-stage optimization; a one-dimensional
search was performed at the outer-stage and SDPs were solved at the
inner-stage.

\subsection{Proposed SOCP-based Solution}

We now solve \eqref{eq:secrecyrate} based on the proposed approximations
presented in the previous section. To do so, we first transform \eqref{eq:secrecyrate}
into an equivalent formulation as \begin{subequations}\label{eq:secrecyrate:reform1}
\begin{align}
\underset{\mathbf{w},\alpha>0,\beta>0}{\mini} & \;\frac{\beta}{\alpha}\\
\st & \;\frac{\mathbf{w}\herm(\mathbf{A}+\mathbf{G})\mathbf{w}+1}{\mathbf{w}\herm\mathbf{G}\mathbf{w}+1}\geq\alpha\label{eq:des:sinr}\\
 & \;\frac{\mathbf{w}\herm(\mathbf{B}_{m}+\mathbf{H}_{m})\mathbf{w}+1}{\mathbf{w}\herm\mathbf{H}_{m}\mathbf{w}+1}\leq\beta,\;\forall m\label{eq:eve:sinr}\\
 & \;\mathbf{w}\herm\mathbf{C}\mathbf{w}\leq P_{\mathrm{tot}},\;\mathbf{w}\herm\D(\mathbf{e}_{k})\mathbf{C}\mathbf{w}\leq P_{k},\forall k\label{eq:secureAF:power}
\end{align}
\end{subequations} In fact, \eqref{eq:secrecyrate:reform1} is an
epigraph form of \eqref{eq:secrecyrate} assuming the optimal value
of the latter is strictly positive. Regarding \eqref{eq:secrecyrate:reform1},
the objective function can be approximated using the upper bound in
\eqref{eq:App3}, while the constraints \eqref{eq:des:sinr} and \eqref{eq:eve:sinr}
can be approximated by \eqref{eq:App4} and \eqref{eq:app6}, respectively.
The resulting convexified subproblem is an SOCP.

To complete the first application, we now provide the worst-case computational
complexity of the proposed SOCP-based method and the SDP-based solution
in \cite{Yang:SecureAF:2013}, using the results in \cite{SOCApp_boyd1999}.
For the former, the worst-case arithmetic cost per iteration is $\mathcal{O}\bigl(K^{3}M(K+M)^{0.5}\bigr)$
which reduces to $\mathcal{O}\bigl(K^{3.5}\bigr)$ when $K\gg M$.
For the latter, the worst-case per-iteration computational cost is
$\mathcal{O}\left(K^{4}(M+K)^{0.5}(M+K^{2})\right)$ reducing to $\mathcal{O}\left(K^{6.5}\right)$
when $K\gg M$.\footnote{We omit the constant related to the desired solution accuracy for
the complexity analysis.} The analysis implies that the per-iteration complexity of the proposed
solution is much less sensitive to $K$ than that in \cite{Yang:SecureAF:2013}.
We recall that the optimization method in \cite{Yang:SecureAF:2013}
is also an iterative procedure. In addition, the complexity of each
subproblem in the proposed method is much less than that in \cite{Yang:SecureAF:2013}
(in orders of magnitude). Thus, we can reasonably expect that the
proposed solution is superior to the SDP-based method in \cite{Yang:SecureAF:2013}
in terms of numerical efficiency, which will be elaborated by numerical
experiments in the following.

\subsection{Numerical Results}

\begin{figure}
\centering
\subfigure[Average secrecy rate (in bps).]{\label{fig:secureafrate}\includegraphics[width=1\columnwidth]{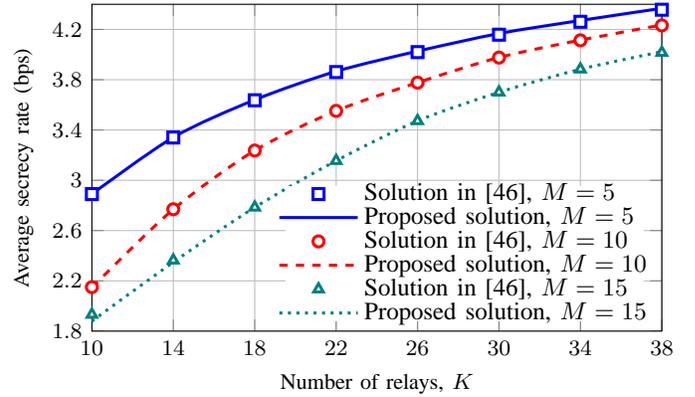}}\\
\subfigure[Average run time (in seconds).]{ \label{fig:secureaftime}\includegraphics[width=1\columnwidth]{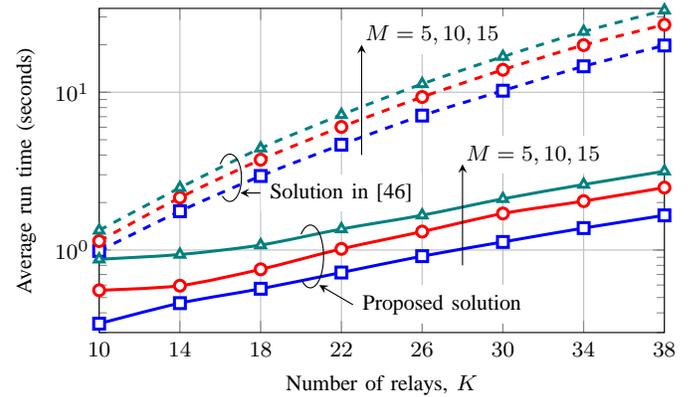}}

\caption{Average secrecy rate and average run time of the proposed solution
and the solution in \cite{Yang:SecureAF:2013} versus different numbers
of relays, $K$. We take the number of eavesdroppers, $M$, as $5$,
$10$, and $15$.}
 \label{fig:secureafperf}
\end{figure}

We now numerically evaluate the performances of the proposed solution
in terms of achieved secrecy rate and computational complexity (i.e.
run time). The considered simulation model follows the one in \cite{Yang:SecureAF:2013}.
Specifically, the channels are independent Rayleigh fading with zero
means and unit variances. The noise variance is set to $\sigma^{2}=1$.
The transmit power at the source is $P_{\mathrm{s}}=20$ dB, the maximum
total transmit power at the relays is $P_{\textrm{tot}}=15$ dB, and
the power budget at antenna $k$ is determined as $P_{k}=P_{\textrm{tot}}/2K$
if $k$ is odd and $P_{k}=2P_{\textrm{tot}}/K$ otherwise. We note
that a feasible $(\mathbf{w},\alpha,\beta)$ can be easily generated
as follows. First, a random (but small enough) $\mathbf{w}$ is generated
satisfying \eqref{eq:secureAF:power}. Then, the left sides of the
constraints \eqref{eq:des:sinr} and \eqref{eq:eve:sinr} are computed
accordingly, from which feasible $\alpha$ and $\beta$ can be found
easily.

In Fig. \ref{fig:secureafrate}, we plot the achieved secrecy rate
(in bits per second) of the considered schemes with different numbers
of relays and eavesdroppers. We can observe that, in all cases, the
proposed SOCP-based solution achieves nearly the same performance
as the SDP-based solution in \cite{Yang:SecureAF:2013}, but with
much lower computation time as shown in Fig. \ref{fig:secureaftime}.
As expected, the run time of the proposed solution increases slowly
with $M$ and $K$ compared to that of the existing method. In particular,
the computation time improvement achieved by the proposed solution
is huge for large numbers of relays and eavesdroppers (approximately
10 times faster in favor of the proposed method at $K=38$). This
observation is consistent with the theoretical bounds of the arithmetical
cost reported in the previous subsection.

\section{Application II: Beamforming Designs for Cognitive Radio Multicasting
\label{sec:Application-II:-Cognitive}}

We now turn our attention to the cognitive radio which has been considered
as one of the most promising techniques to improve the spectrum utilization.
In particular, we revisit the problem of transmit power minimization
for secondary multicasting investigated in \cite{Anh:Cogni:2012}.

\subsection{System Model and Problem Formulation}

The considered system model consists of a secondary multi-antenna
base station transmitting data to $G$ groups of secondary single-antenna
users where users in the same group receive the same information content.
Let $\mathcal{U}_{g}$, $g=1,...,G$, denote the set of users in group
$g$. The total number of secondary users is $M$, and each user belongs
to only one group. In addition, there exist $L$ primary single-antenna
users who are interfered by the secondary transmission. A diagram
of the considered system is shown in Fig. \ref{fig:Cogmulticastmodel}.
Let $N$ be the number of antennas equipped at the secondary BS, $\mathbf{h}_{i}\in\mathbb{C}^{1\times N}$
be the channel (row) vector between the secondary BS and secondary
user (SU) $i$, $\mathbf{l}_{j}\in\mathbb{C}^{1\times N}$ be the
channel vector between the secondary BS and primary user (PU) $j$,
and $\mathbf{x}_{g}\in\mathbb{C}^{N\times1}$ be the multicast transmit
beamforming vector at the secondary BS for group $g$. The problem
of minimizing the transmit power at the secondary BS is stated as
\cite{Anh:Cogni:2012}\begin{subequations}\label{eq:cogmulticast}
\begin{align}
\underset{\{\mathbf{x}_{g}\}_{g=1}^{G}}{\mini} & \;\sum_{g=1}^{G}||\mathbf{x}_{g}||_{2}^{2}\label{eq:cogmulticast-obj}\\
\st & \;\frac{|\mathbf{h}_{i}\mathbf{x}_{g}|^{2}}{\sum_{k\neq g}|\mathbf{h}_{i}\mathbf{x}_{k}|^{2}+\sigma_{i}^{2}}\geq\alpha_{i},\forall i\in\mathcal{U}_{g},g=1,...,G\label{eq:cogmulticast_secondary}\\
 & \;\sum_{g=1}^{G}|\mathbf{l}_{j}\mathbf{x}_{g}|^{2}\leq\beta_{j},\;\forall j=1,...,L\label{eq:cogmulticast_primary}
\end{align}
\end{subequations} where $\beta_{j}$ is the predefined interference
threshold at PU $j$, $\alpha_{i}$ and $\sigma_{i}^{2}$ are the
QoS level and noise variance corresponding to SU $i$, respectively.
The constraints in \eqref{eq:cogmulticast_secondary} guarantee the
QoSs of the SUs while those in \eqref{eq:cogmulticast_primary} ensure
that the interference generated by the secondary transmission at the
PUs are smaller than predefined thresholds.
\begin{figure}
\centering{}\includegraphics[width=1\columnwidth]{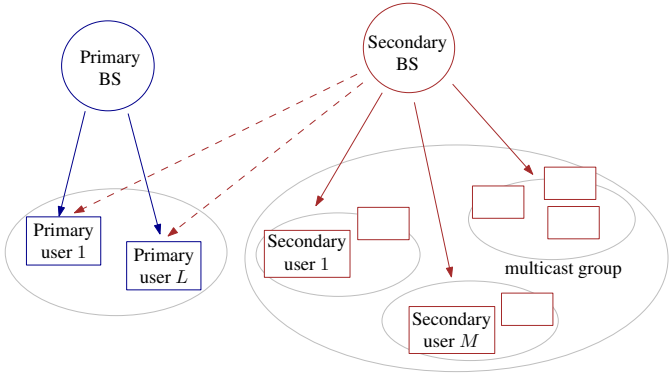}\caption{An example of cognitive multicast transmissions with $L$ PUs and
multiple multicast groups including total $M$ SUs.}
 \label{fig:Cogmulticastmodel}
\end{figure}

Problem \eqref{eq:cogmulticast} is a nonconvex program, and the prevailing
approach is to lift the problem into a SDP \cite{Multical:Luo:2006,multicast_zlou_2008},
i.e. PSD matrices $\mathbf{X}_{g}=\mathbf{x}_{g}\mathbf{x}_{g}\herm$,
$\forall g$, are introduced and the rank-1 constraints on $\{\mathbf{X}_{g}\}_{g=1}^{G}$
are ignored. However, this approach cannot guarantee even a feasible
solution, since the relaxed problem generally does not yield rank-1
solutions and randomization procedures are inefficient. To overcome
this shortcoming, \cite{Anh:Cogni:2012} proposed an iterative approach
where the PSD matrices $\{\mathbf{X}_{g}\}_{g=1}^{G}$ were still
introduced. However, the noncovex rank-1 constraints (which were expressed
a reverse convex constraint, i.e. $\sum_{g=1}^{G}\tr(\mathbf{X}_{g})-\lambda_{\max}(\mathbf{X}_{g})\leq0$
where $\lambda_{\max}(\mathbf{X}_{g})$ is the maximal eigenvalue
of $\mathbf{X}_{g}$) were sequentially approximated in a manner similar
to the SCA framework presented in this paper.

\subsection{Proposed SOCP-based Solution}

We observe that the objective function in \eqref{eq:cogmulticast-obj}
and constraint in \eqref{eq:cogmulticast_primary} are convex. Thus
the difficulty of solving the problem comes from the nonconvex constraints
in \eqref{eq:cogmulticast_secondary}. To handle these constraints,
let us introduce $\mathbf{x}\triangleq[\mathbf{x}_{1};...;\mathbf{x}_{G}]$
and equivalently rewrite \eqref{eq:cogmulticast} as \begin{subequations}\label{eq:cogmulticast-1}
\begin{align}
\underset{\mathbf{x}}{\mini} & \;||\mathbf{x}||_{2}^{2}\label{eq:cogmulticast_obj}\\
\st & \;\frac{\mathbf{x}\herm\hat{\mathbf{H}}_{i}\mathbf{x}}{\mathbf{x}\herm\bar{\mathbf{H}}_{i}\mathbf{x}+\sigma_{i}^{2}}\geq\alpha_{i},\forall i\in\mathcal{U}_{g},g=1,...,G\label{eq:cogmulticast_constr1}\\
 & \;\mathbf{x}\herm\mathbf{L}_{j}\mathbf{x}\leq\beta_{j},\;\forall j\label{eq:cogmulticast_constr2}
\end{align}
\end{subequations} where $\hat{\mathbf{H}}_{i}=\blkdiag(\boldsymbol{0},\mathbf{H}_{i},\boldsymbol{0})$,
$\mathbf{H}_{i}\triangleq\mathbf{h}_{i}\herm\mathbf{h}_{i}$, $\bar{\mathbf{H}}_{i}=\blkdiag(\underbrace{\mathbf{H}_{i},\ldots\mathbf{H}_{i}}_{g-1\textrm{ terms}},\mathbf{0},\underbrace{\mathbf{H}_{i},\ldots\mathbf{H}_{i}}_{G-g\textrm{ terms}})$
(for $i\in\mathcal{U}_{g}$), and $\mathbf{L}_{j}=\blkdiag(\underbrace{\mathbf{l}_{j}\herm\mathbf{l}_{j},\ldots\mathbf{l}_{j}\herm\mathbf{l}_{j}}_{G\textrm{ terms}})$.
Now, we can easily use \eqref{eq:App4} to deal with \eqref{eq:cogmulticast_constr1},
and the convex approximated problem is actually an SOCP. We remark
that, for \eqref{eq:cogmulticast}, it is nontrivial to find a feasible
$\{\mathbf{x}_{g}\}_{g=1}^{G}$ to start the SCA procedure. Thus the
relaxed version, i.e. Algorithm \ref{alg:generate-initial-SCA}, is
invoked for some first iterations, until a feasible point is found.

We now compare the complexity of the proposed solution and the method
in \cite{Anh:Cogni:2012}. In particular, the worst-case complexity
for solving the SOCP in \eqref{eq:cogmulticast-1} is $\mathcal{O}\left(G^{3}N^{3}(M+L)^{1.5}\right)$,
while that for the SDP in \cite{Anh:Cogni:2012} is $\mathcal{O}\left(G^{2}N^{4}(M+L+GN)^{0.5}(M+L+GN^{2})\right)$.
When the number of transmit antennas at the secondary BS is large,
the two bounds reduce to $\mathcal{O}\left(N^{3}\right)$ and $\mathcal{O}\left(N^{6.5}\right)$,
respectively.

\subsection{Numerical Results}

\begin{figure}
\centering
\subfigure[Average transmit power at the secondary BS (in dB).]{\label{fig:multicastpower}\includegraphics[width=0.95\columnwidth]{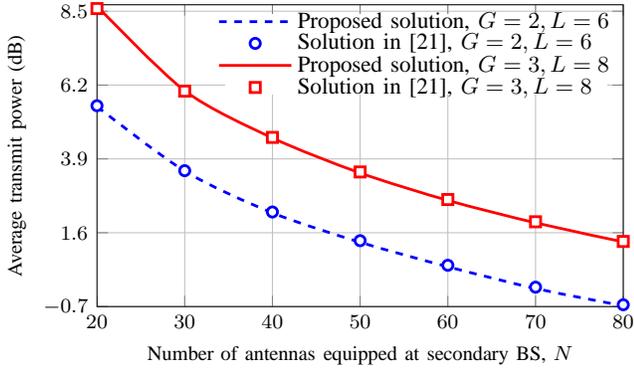}}\\
\subfigure[Average run time (in seconds).]{ \label{fig:multicasttime}\includegraphics[width=0.95\columnwidth]{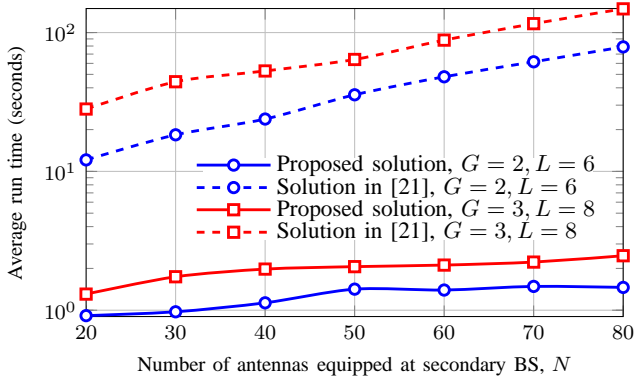}}

\caption{Average transmit power at the secondary BS and average run time of
the proposed solution and the solution in \cite{Anh:Cogni:2012} versus
the number of equipped antennas, $N$. Two sets of $(G,L)$ are taken
those are $(2,6)$ and $(3,8)$. Each multicast group includes 4 secondary
users. }
 \label{fig:multicastperf}
\end{figure}

We follow the simulation model considered in \cite{Anh:Cogni:2012}
for performance comparisons. Particularly, the channels are generated
as $\mathbf{h}_{i}\thicksim\mathcal{CN}(0,\mathbf{I}_{N})$, $\mathbf{l}_{j}\thicksim\mathcal{CN}(0,\mathbf{I}_{N})$.
The QoS levels at the SUs and the interference thresholds at the PUs
are set to $\alpha_{i}=10\;\textrm{dB}$, $\forall i$ and $\beta_{j}=5\;\textrm{dB},$
$\forall j$. The stopping criterion of the method in \cite{Anh:Cogni:2012}
is when the decrease in the last 5 iterations is smaller than $10^{-3}$.

In Fig. \ref{fig:multicastpower}, we plot the average required transmit
power at the secondary BS for the proposed method and the one in \cite{Anh:Cogni:2012}
as functions of the number of transmit antennas $N$. As can be seen,
the transmit powers required by the two schemes are almost the same
for all considered scenarios. In other words, the proposed solution
is similar to the one in \cite{Anh:Cogni:2012} in terms of power
efficiency. However, the proposed method is much more computationally
efficient as demonstrated in Fig. \ref{fig:multicasttime}, in which
we plot the run time of the two methods. We can clearly see that the
average run time of both schemes increases with $N$, but slowly for
the proposed solution and very rapidly for the method in \cite{Anh:Cogni:2012}.
As a result, the computation time of the method in \cite{Anh:Cogni:2012}
is much higher than that of the proposed solution, especially for
large $N$. This numerical observation is consistent with the complexity
analysis provided in the preceding subsection. In summary, the numerical
results in Fig. \ref{fig:multicastperf} demonstrate that the proposed
SOCP-based solution is superior to the existing one presented in \cite{Anh:Cogni:2012}.

\section{Application III: Precoding Design for MIMO Relaying \label{sec:Application-III:-MIMO}}

Relay-assisted wireless communications is expected to play a key role
in improving coverage and spectral efficiency for the current and
future generations of cellular networks. In this section, we apply
the proposed approximations to the scenario of multiuser MIMO relaying
which was investigated in \cite{Anh:MIMOdc:2013}.
\begin{figure}
\centering{}\includegraphics[width=0.9\columnwidth]{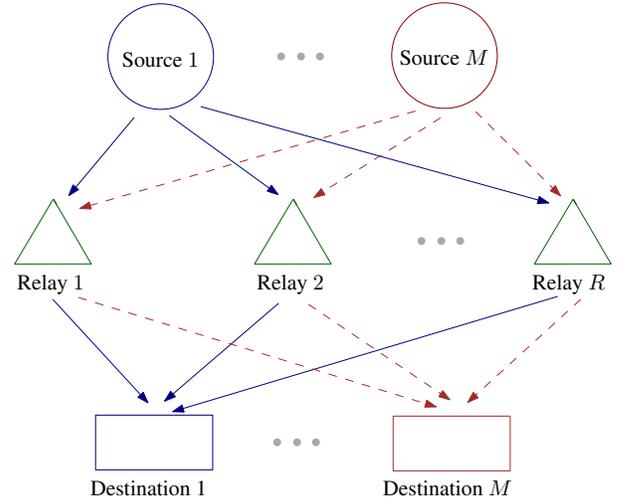}\caption{An example of relaying communications with $M$ source-destination
pairs and $R$ MIMO relays. The transmissions follow the two-phase
precode-and-forward principle.}
 \label{fig:MIMOrelaymodel}
\end{figure}

\subsection{System Model and Problem Formulation}

We consider the wireless communication scenario in which the transmission
of $M$ source-destination pairs are simultaneously assisted by a
set of $R$ relays. Each of the sources and destinations is equipped
with single antenna, whereas each relay is equipped with $N_{\mathrm{R}}$
antennas (the total number of antennas at the relays is $N=RN_{\mathrm{R}}$).
Suppose there are no direct links between the sources and destinations,
and the relays operate according to the AF protocol. That is, the
sources transmit their information to the relays in the first phase,
and then the relays process the received signals and retransmit them
to the destinations in the second phase. The considered system model
is illustrated in Fig. \ref{fig:MIMOrelaymodel}. We reuse the notations
introduced in \cite{Anh:MIMOdc:2013}. Specifically, let $\mathbf{s}=[s_{1},s_{2},...,s_{M}]\trans\in\mathbb{C}^{M\times1}$,
where $E[|s_{i}|^{2}]=\sigma_{\mathrm{s}}^{2}$, be the vector of
messages sent by the sources, $\mathbf{h}_{i}=[h_{i1},h_{i2},...,h_{iN}]\trans\in\mathbb{C}^{N\times1}$
and $\mathbf{l}_{i}=[l_{i1},l_{i2},...,l_{iN}]\trans\in\mathbb{C}^{N\times1}$
be the vectors of channels from source $i$ to the relays and from
the relays to destination $i$, respectively. Let $\mathbf{X}_{k}\in\mathbb{C}^{N_{\mathrm{R}}\times N_{\mathrm{R}}}$
be the processing matrix at relay $k$, $k=1,...,R$. For ease of
description, we also define $\mathbf{H}\triangleq[\mathbf{h}_{1},\mathbf{h}_{2},...,\mathbf{h}_{M}]$,
$\mathbf{L}\triangleq[\mathbf{l}_{1},\mathbf{l}_{2},...,\mathbf{l}_{M}]$,
and $\mathbf{X}\triangleq\mathrm{blkdiag}\{\mathbf{X}_{k}\}_{k}$
as in \cite{Anh:MIMOdc:2013}. Then the signal vectors received at
the relays and destination $i$ are given by
\begin{equation}
\mathbf{y}_{\mathrm{up}}=\mathbf{Hs}+\mathbf{n}_{\mathrm{re}}
\end{equation}
and
\begin{equation}
y_{\mathrm{de},i}=\mathbf{l}_{i}\trans\mathbf{X}\mathbf{h}_{i}s_{i}+\sum_{j\neq i}\mathbf{l}_{i}\trans\mathbf{X}\mathbf{h}_{j}s_{j}+\mathbf{l}_{i}\trans\mathbf{X}\mathbf{n}_{\mathrm{re}}+n_{\mathrm{de},i}
\end{equation}
respectively, where $\mathbf{n}_{\mathrm{re}}\thicksim\mathcal{CN}(0,\sigma_{\mathrm{re}}^{2}\mathbf{I}_{N})$
and $n_{\mathrm{de},i}\thicksim\mathcal{CN}(0,\sigma_{\mathrm{de}}^{2})$
are the noise vectors at the relays and destination $i$. Accordingly,
the SINR at the $i$th destination is written as
\begin{equation}
\gamma_{i}(\mathbf{X})=\frac{\sigma_{\mathrm{s}}^{2}|\mathbf{l}_{i}\herm\mathbf{X}\mathbf{h}_{i}|^{2}}{\sigma_{\mathrm{s}}^{2}\sum_{j\neq i}|\mathbf{l}_{i}\herm\mathbf{X}\mathbf{h}_{j}|^{2}+\sigma_{\mathrm{re}}^{2}\left\Vert \mathbf{X}\herm\mathbf{l}_{i}\right\Vert _{2}^{2}+\sigma_{\mathrm{de}}^{2}}\label{eq:SINR:Relay}
\end{equation}
The power consumption at antenna $n$, $n=1,...,N$, is given by
\begin{equation}
P_{n}=\mathbf{e}_{n}\mathbf{X}\mathbf{B}\mathbf{X}\herm\mathbf{e}_{n}\herm
\end{equation}
where $\mathbf{B}\triangleq\sigma_{\mathrm{s}}^{2}\mathbf{H}\mathbf{H}\herm+\sigma_{\mathrm{re}}^{2}\mathbf{I}$,
and $\mathbf{e}_{n}\triangleq[\underset{n-1}{\underbrace{0,\ldots,0}},1,\underset{N-n}{\underbrace{0,\ldots,0}}]$.
Based on the above discussions, the problem of max-min fairness data
rate between the source-destination pairs under the power constraint
of each individual antenna at the relays reads \cite{Anh:MIMOdc:2013}
\begin{subequations}\label{eq:MIMOrelay}
\begin{align}
\underset{\mathbf{X}}{\maxi} & \;\underset{1\leq i\leq M}{\min}\;\gamma_{i}(\mathbf{X})\\
\st & \;\mathbf{e}_{n}\mathbf{X}\mathbf{B}\mathbf{X}\herm\mathbf{e}_{n}\herm\leq\bar{P}_{n},\;\forall n=1,...,N\label{eq:MIMOrelay_PAPC}
\end{align}
\end{subequations} where $\bar{P}_{n}$ is the maximum transmit power
at antenna $n$. The left sides of constraints in \eqref{eq:MIMOrelay_PAPC}
are convex quadratic functions, since matrix $\mathbf{B}$ is positive
definite. Thus the feasible set of \eqref{eq:MIMOrelay} is convex.
However, \eqref{eq:MIMOrelay} is still intractable because $\gamma_{i}(\mathbf{X})$
is neither concave nor convex with respect to $\mathbf{X}$. To solve
\eqref{eq:MIMOrelay}, \cite{Anh:MIMOdc:2013} transformed it to a
DC program where the resulting feasible set contains linear matrix
inequalities (LMIs) and the objective is a DC function. After that,
\cite{Anh:MIMOdc:2013} applied the DCA to solve the DC program which
results in an SDP in each iteration.

\subsection{Proposed SOCP-based Solution}

To apply the proposed approximations, let us define $\mathbf{x}\triangleq\myvec(\mathbf{X})\in\mathbb{C}^{R^{2}N_{\mathrm{R}}^{2}}$
and rewrite the expression in \eqref{eq:SINR:Relay} with respect
to $\mathbf{x}$. To this end we recall two useful inequalities: $\tr(\mathbf{A}\herm\mathbf{B})=\myvec(\mathbf{A})\herm\myvec(\mathbf{B})$
and $\myvec(\mathbf{A}\mathbf{B})=(\mathbf{I}_{n}\otimes\mathbf{A})\myvec(\mathbf{B})$.
Now it is easy to see that \eqref{eq:SINR:Relay} is equivalent to
\begin{align}
\gamma_{i}(\mathbf{x}) & =\frac{\sigma_{\mathrm{s}}^{2}\mathbf{x}\herm\bar{\mathbf{H}}_{ii}\mathbf{x}}{\sigma_{\mathrm{s}}^{2}\sum_{j\neq i}\mathbf{x}\herm\bar{\mathbf{H}}_{ij}\mathbf{x}+\sigma_{\mathrm{re}}^{2}\mathbf{x}\herm\tilde{\mathbf{L}}_{i}\mathbf{x}+\sigma_{\mathrm{de}}^{2}}\nonumber \\
 & =\frac{\sigma_{\mathrm{s}}^{2}\mathbf{x}\herm\bar{\mathbf{H}}_{ii}\mathbf{x}}{\mathbf{x}\herm\tilde{\mathbf{H}}_{i}\mathbf{x}+\sigma_{\mathrm{de}}^{2}}
\end{align}
where $\bar{\mathbf{H}}_{ij}\triangleq\myvec(\mathbf{l}_{i}\mathbf{h}_{j}\herm)\myvec(\mathbf{l}_{i}\mathbf{h}_{j}\herm)\herm$,
$\tilde{\mathbf{L}}_{i}\triangleq(\mathbf{I}_{N}\otimes\mathbf{l}_{i}\herm)\herm(\mathbf{I}_{N}\otimes\mathbf{l}_{i}\herm)$,
and $\tilde{\mathbf{H}}_{i}\triangleq\sigma_{\mathrm{s}}^{2}\sum_{j\neq i}\bar{\mathbf{H}}_{ij}+\sigma_{\mathrm{re}}^{2}\tilde{\mathbf{L}}_{i}$.
Based on above discussion, we can reformulate \eqref{eq:MIMOrelay}
using its epigraph form as \begin{subequations}\label{eq:mimorelay-epigraph}
\begin{align}
\underset{\mathbf{x},\vartheta>0}{\maxi} & \;\vartheta\\
\st & \;\gamma_{i}(\mathbf{x})\geq\vartheta,\;\forall i=1,...,M\label{eq:MIMOrelayconstr1-epi}\\
 & \;\mathbf{x}\herm\tilde{\mathbf{B}}_{n}\mathbf{x}\leq\bar{P}_{n},\;\forall n=1,...,N.\label{eq:MIMOrelayconstr3-epi}
\end{align}
\end{subequations} where $\tilde{\mathbf{B}}_{n}\triangleq(\mathbf{I}_{N}\otimes\mathbf{A}_{n})\herm(\mathbf{B}\trans\otimes\mathbf{I}_{N})$,
$\mathbf{A}_{n}\triangleq\mathbf{e}_{n}\mathbf{e}_{n}\herm$. Now,
the nonconvex constraints in \eqref{eq:MIMOrelayconstr1-epi} can
be straightforwardly approximated by \eqref{eq:App5}.

The complexity analysis is as follows. We recall that the dimension
of $\mathbf{x}$ in \eqref{eq:mimorelay-epigraph} is $R^{2}N_{\mathrm{R}}^{2}$,
however the number of complex variables in $\mathbf{x}$ is only $RN_{\mathrm{R}}^{2}$
(other elements are zero according to the arrangement of $\mathbf{X}$).
Therefore, assumming $RN_{\mathrm{R}}\gg M$, the worst-case complexity
for solving \eqref{eq:mimorelay-epigraph} is $\mathcal{O}(R^{4.5}N_{\mathrm{R}}^{6.5})$.
On the other hand, for the SDP considered in \cite{Anh:MIMOdc:2013},
the corresponding worst-case complexity is $\mathcal{O}(R^{6}N_{\mathrm{R}}^{8})$.
Thus, the proposed SOCP-based approach achieves significant gains
in terms of computation time over the SDP-based method in \cite{Anh:MIMOdc:2013}.

\subsection{Numerical Results }

\begin{figure}
\centering
\subfigure[Average achieved minimum SINR (in dB).]{\label{fig:relaymimosnr}\includegraphics[width=0.95\columnwidth]{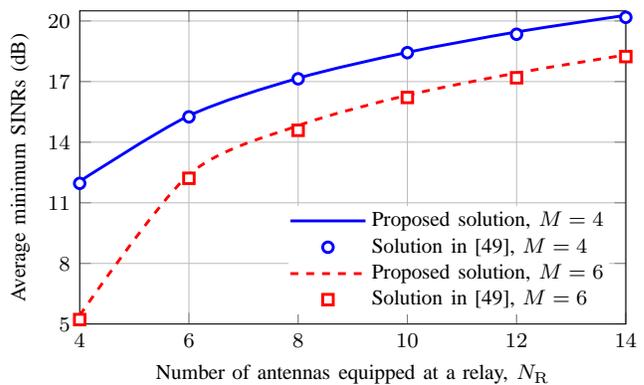}}\\
\subfigure[Average run time (in seconds).]{ \label{fig:relaymimotime}\includegraphics[width=0.95\columnwidth]{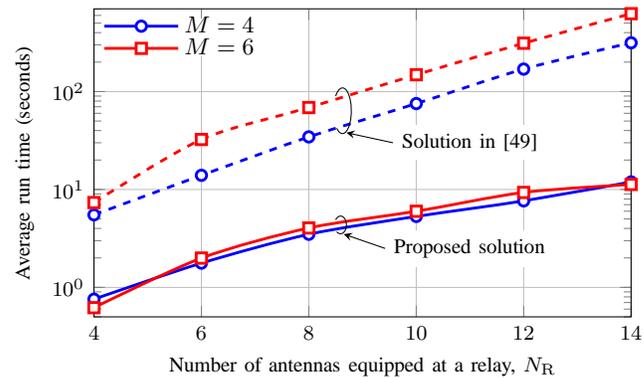}}

\caption{Average achieved minimum SINR and the computational costs (in terms
of run time) of the proposed solution and the solution in \cite{Anh:MIMOdc:2013}
versus number of antennas equipped at a relay, $N_{\mathrm{R}}$.
The number of source-destination pairs $M$ is taken as $4$ and $6$.}
 \label{fig:relaymimoperf}
\end{figure}

We evaluate the proposed solution in terms of run time and achieved
minimum SINR using the system model studied in \cite{Anh:MIMOdc:2013}.
The channels are randomly generated as $\mathbf{h}_{i}\thicksim\mathcal{CN}(0,\mathbf{I}_{N})$,
and $\mathbf{l}_{i}\thicksim\mathcal{CN}(0,\mathbf{I}_{N})$. The
noise variances at the relays and destinations are set to $\sigma_{\mathrm{re}}^{2}=\sigma_{\mathrm{de}}^{2}=1$.
The transmit power is $\sigma_{\mathrm{s}}^{2}=20$ dB for all the
sources and the number of relays is $R=2$. We simply set the power
budget $P_{\mathrm{R}}=10$ dB for all relays, and the maximum transmit
power at each individual antenna is $\bar{P}_{n}=P_{\mathrm{R}}/N_{\mathrm{R}}$,
$\forall n$. For this problem, a feasible initial point to \eqref{eq:mimorelay-epigraph}
can be easily created by simple manipulations. That is, $\mathbf{x}^{(0)}$
is randomly generated and then rescaled, if required, to satisfy the
power constraints in \eqref{eq:MIMOrelay_PAPC}. The iterative method
in \cite{Anh:MIMOdc:2013} terminates if the increase in the objective
in the last 5 iterations is less than $10^{-3}$.

In Fig. \ref{fig:relaymimosnr} we compare the achieved minimum SINR
of the proposed SOCP-based method and the SDP-based solution in \cite{Anh:MIMOdc:2013}
as functions of $N_{\mathrm{R}}$. We can observe that the minimum
SINR performance of two methods of comparison is the same for all
considered scenarios. However the method in \cite{Anh:MIMOdc:2013}
comes at the cost of much higher computational complexity which is
shown in Fig. \ref{fig:relaymimotime}. In particular, Fig. \ref{fig:relaymimotime}
shows that the complexity of both methods increases with respect to
$N_{\mathrm{R}}$. We can also see that the run time of the proposed
solution is much smaller than that of the SDP-based solution. This
observation agrees with the complexity analysis presented in the previous
subsection.

\section{Application IV: Power Management for Weighted Sum Rate Maximization
and Max-mix Fairness Energy Efficiency in Multiuser Multicarrier Systems
\label{sec:Application-IV:-Power}}

In the last application, we apply the proposed approximations to solve
the power control problems for maximizing weighted sum rate and energy
efficiency fairness. While the former is a classical problem, the
latter has been receiving growing attention in recent years for green
wireless networks. The power control problem for weighted sum rate
maximization (WSRmax) was studied in \cite{Wang:SCALE:2012}, using
a GP method, while the one for energy efficiency fairness was recently
studied in \cite{Zappone:EEpowercontrol:2015} using a generic NLP
approach. We will show in the following that these two important problems
can be simply solved by the proposed conic quadratic formulations.

\subsection{System Model and Problem Formulation}

\begin{figure}
\centering \includegraphics[width=0.95\columnwidth]{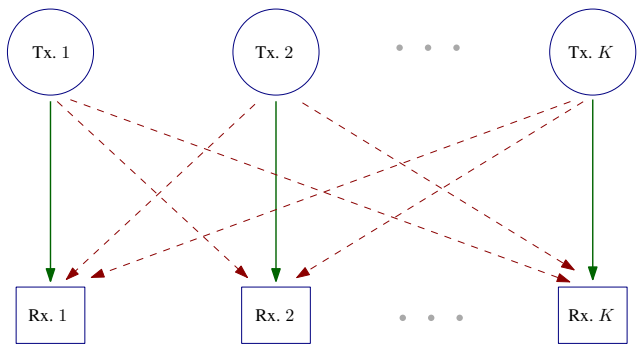}

\caption{A diagram of multiuser multicarrier systems. The transmitters and
receivers are labeled as $\textrm{Tx.}$ and $\textrm{Rx}.$, respectively.
The green (solid) lines represent the desire signals, while the red
(dashed) ones denote the interference from other transmitters. The
lines in the figure represent one of $N$ independent subcarriers.}

\label{fgi:mumc-1}
\end{figure}

Consider a general network of $K$ communication links where each
of the transmitters and receivers is equipped with a single antenna
as illustrated in Fig. \ref{fgi:mumc-1}. This system model has been
the subject in many works such as \cite{Wang:SCALE:2012,Zappone:EEpowercontrol:2015,Kha:dcprog:2012,maxminEE:SCA:TVT:2015,scale:2009}.
For the ease of discussion, we follow the notations used in \cite{Wang:SCALE:2012}.
The $k$th transmitter sends its data to the $k$th receiver through
$N$ independent subcarriers. The channel gain between transmitter
$l$ and receiver $k$ over carrier $n$ is denoted by $G_{kln}$.
We note that this setup is a generalization of the system model studied
in \cite{Kha:dcprog:2012,maxminEE:SCA:TVT:2015} where $N=1$. The
interference at each receiver is treated as background noise. Accordingly,
the achievable data rate of the $k$th link is given by
\begin{equation}
R_{k}=\sum_{n=1}^{N}\log\biggl(1+\frac{G_{kkn}p_{kn}}{\sigma^{2}+\sum_{l\neq k}G_{kln}p_{ln}}\biggr)
\end{equation}
where $p_{kn}$ is the amount of power allocated to the $n$th carrier
at the $k$th transmitter, and $\sigma^{2}$ denotes the noise variance.

\subsubsection{Weighted Sum Rate Maximization}

The WSRmax problem is given by \cite{Wang:SCALE:2012,scale:2009}\begin{subequations}\label{eq:PCWSR:origin}
\begin{align}
\underset{\mathbf{p}}{\maxi} & \;\sum_{k=1}^{K}w_{k}R_{k}\\
\st & \;\sum_{n=1}^{N}p_{kn}\leq\bar{p}_{k},\,\forall k;\;0\leq p_{kn}\leq\bar{p}_{kn},\,\forall k,n.\label{eq:mcpowerconstraint}
\end{align}
\end{subequations} where $w_{k}$ is the weight associated with the
$k$th user, $\bar{p}_{k}$ is the maximum transmit power at transmitter
$k$ and $\bar{p}_{kn}$ is the power spectrum mask. It is known that
\eqref{eq:PCWSR:origin} is NP-hard \cite{tomLuocomplexity2008},
and \cite{Wang:SCALE:2012} and \cite{scale:2009} applied an SCA
approach to solve this problem. In particular, \cite{Wang:SCALE:2012}
approximated the WSRmax problem as a sequence of GPs. As discussed
in \cite{Wang:SCALE:2012} the GP-based method suffers some numerical
difficulties.

\subsubsection{Energy Efficiency Fairness }

The problem of max-min energy efficiency fairness (maxminEEfair) for
the multiuser multicarrier system has been considered in many recent
works such as \cite{Zappone:EEpowercontrol:2015,maxminEE:SCA:TVT:2015}
which is stated as \begin{subequations}\label{eq:mumc:EEfair}
\begin{align}
\underset{\mathbf{p}}{\maxi} & \;\underset{1\leq k\leq K}{\min}\;\frac{R_{k}}{\sum_{n=1}^{N}p_{kn}+p_{k}^{\mathrm{c}}}\\
\st & \;\sum_{n=1}^{N}p_{kn}\leq\bar{p}_{k},\,\forall k;\;0\leq p_{kn}\leq\bar{p}_{kn},\,\forall k,n.
\end{align}
\end{subequations} where $p_{k}^{\mathrm{c}}$ is a constant that
denotes the total circuit power consumption for the link between the
$k$th transmitter and the $k$th receiver. To solve \eqref{eq:mumc:EEfair},
\cite{Zappone:EEpowercontrol:2015} proposed a two-stage iterative
method which is a combination of modified Dinkelbach's algorithm and
the SCA technique. More specifically, the SCA framework is applied
to arrive at a max-min concave-convex fractional program which is
solved by the modified Dinkelbach method. We note that the work of
\cite{Zappone:EEpowercontrol:2015} used the same approximations introduced
in \cite{Wang:SCALE:2012,scale:2009,power_controlGP} to deal with
the nonconvexity of $R_{k}$. However, the convex problem achieved
at each iteration of their proposed iterative method is a generic
NLP.

\subsection{Proposed SOCP-based Solutions }

\subsubsection{Proposed Solution for WSRmax\label{subsec:WSRmaxmulticarr}}

Using the epigraph form we can equivalently rewrite \eqref{eq:PCWSR:origin}
as \begin{subequations}\label{eq:PCWSR}
\begin{align}
\underset{\mathbf{p},\mathbf{t}\geq1}{\maxi} & \;\sum_{k=1}^{K}w_{k}\sum_{n=1}^{N}\log t_{kn}\label{eq:PCWSR:obj}\\
\st & \;\frac{\sigma_{kn}^{2}+\sum_{l=1}^{K}G_{kln}p_{ln}}{\sigma_{kn}^{2}+\sum_{l\neq k}G_{kln}p_{ln}}\geq t_{kn},\,\forall k,n\label{eq:PCWSR:softSINR}\\
 & \sum_{n=1}^{N}p_{kn}\leq\bar{p}_{k},\,\forall k;\;0\leq p_{kn}\leq\bar{p}_{kn},\,\forall k,n.
\end{align}
\end{subequations} It is obvious that \eqref{eq:App1} and \eqref{eq:App2}
can be used to handle \eqref{eq:PCWSR:softSINR}, but this does not
lead to an SOCP immediately due to the fact that the cost function
in \eqref{eq:PCWSR:obj} is not in a quadratic form at hand. However,
if $w_{k}=1$ for all $k$ (i.e., the sum rate maximization problem),
then the objective in \eqref{eq:PCWSR:obj} can be replaced by the
geometric mean of $t_{kn}$'s, which is SOC representable \cite{socp_Alizadeh},
and the resulting problem becomes an SOCP. For the general case of
$w_{k}$'s, we can rewrite \eqref{eq:PCWSR:origin} as \begin{subequations}\label{eq:PCWSR:gen}
\begin{align}
\underset{\mathbf{p},\mathbf{t}\geq0}{\maxi} & \;\sum_{k=1}^{K}w_{k}\sum_{n=1}^{N}t_{kn}\label{eq:PCWSR:obj-1}\\
\st & \;u_{kn}\log\frac{u_{kn}}{\tilde{u}_{kn}}\geq u_{kn}t_{kn},\,\forall k,n\label{eq:PCWSR:softSINR-1}\\
 & \sum_{n=1}^{N}p_{kn}\leq\bar{p}_{k},\,\forall k;\;0\leq p_{kn}\leq\bar{p}_{kn},\,\forall k,n.
\end{align}
\end{subequations} where $u_{kn}\triangleq\sigma_{kn}^{2}+\sum_{l=1}^{K}G_{kln}p_{ln}$
and $\tilde{u}_{kn}\triangleq\sigma_{kn}^{2}+\sum_{l\neq k}G_{kln}p_{ln}$.
Note that $u_{kn}$ and $\tilde{u}_{kn}$ are not newly introduced
variables but affine expressions of $p_{kn}$'s. We remark that $u_{kn}\log\frac{u_{kn}}{\tilde{u}_{kn}}$
is jointly convex with $u_{kn}$ and $\tilde{u}_{kn}$. Thus, in light
of the SCA framework, we can approximate the left hand side in \eqref{eq:PCWSR:softSINR-1}
using the first order and the right hand side using \eqref{eq:App1}
or \eqref{eq:App2}, which results in an SOCP.

The proposed method described above introduces $KN$ additional auxiliary
variables, which has a detrimental impact on the overall complexity
for large-scale problems. For \eqref{eq:PCWSR}, we are able to arrive
at a more efficient formulation. The idea is that since the feasible
set of \eqref{eq:PCWSR:origin} is defined by linear constraints,
we can approximate the objective by a quadratic function based on
the Lipschitz continuity to arrive at a QP. This approximation method
has been briefly discussed in Section \ref{subsec:Choice-of-Approximation}
and will be elaborated in the following. Let us rewrite \eqref{eq:PCWSR:origin}
as \begin{subequations}\label{eq:PCWSR:DC}
\begin{align}
\underset{\mathbf{p},\mathbf{t}}{\mini} & \;f(\mathbf{p},\mathbf{t})\triangleq-\sum_{k=1}^{K}w_{k}\sum_{n=1}^{N}\log\Bigl(1+\frac{p_{kn}}{t_{kn}}\Bigr)\\
\st & \;t_{kn}=\tilde{\sigma}_{kn}^{2}+\sum_{l\neq k}\tilde{G}_{kln}p_{ln},\,\forall k,n;\eqref{eq:mcpowerconstraint}
\end{align}
\end{subequations} where $\mathbf{t}\triangleq\{t_{kn}\}_{k,n}$,
$\tilde{\sigma}_{kn}^{2}=\sigma^{2}/G_{kkn}$ and $\tilde{G}_{kln}=G_{kln}/G_{kkn}$
for $l\ne k$. For mathematical exposition, we temporarily treat $\mathbf{t}$
as a vector of newly introduced variables. It implicitly holds that
$0<\tilde{\sigma}_{kn}^{2}\leq t_{kn}\leq\tilde{\sigma}_{kn}^{2}+\sum_{l\neq k}\tilde{G}_{kln}\bar{p}_{ln}$,
$\forall k,n$. In particular, we have the following lemma.
\begin{lem}
\label{lem:lemma1}Consider the function $f(x,y)=\log(1+\frac{x}{y})$
over the domain $D=\{(x,y)|x\geq x_{0}\geq0,y\geq y_{0}>0\}$. Then
the Hessian of $f(x,y)$ satisfies $\nabla^{2}f(x,y)+\rho\mathbf{I}_{2}\succeq\boldsymbol{0}$
for $\rho\geq\rho_{0}=(x_{0}+y_{0})^{-2}$.
\end{lem}
\begin{proof}
The Hessian of $f(x,y)$ is given by
\begin{multline}
\nabla^{2}f(x,y)=\left[\begin{array}{cc}
-\frac{1}{(x+y)^{2}} & \frac{-1}{(x+y)^{2}}\\
\frac{-1}{(x+y)^{2}} & \frac{x(x+2y)}{(xy+y^{2})^{2}}
\end{array}\right]\\
=\frac{1}{(x+y)^{2}}\left[\begin{array}{cc}
-1 & -1\\
-1 & \frac{x(x+2y)}{y}
\end{array}\right]
\end{multline}
It is easy to see that since $x\geq x_{0}$ and $y>y_{0}$, $\nabla^{2}f(x,y)+\rho\mathbf{I}_{2}\succeq\boldsymbol{0}$
for $\rho\geq\rho_{0}(x_{0}+y_{0})^{-2}$, which completes the proof.
\end{proof}
Lemma \ref{lem:lemma1} implies that the function $\log(1+\frac{x}{y})+\rho x^{2}+\rho y^{2}$
is convex for $\rho\geq\rho_{0}$. As a result, we can rewrite the
objective function in \eqref{eq:PCWSR:DC} as the following DC function
\begin{equation}
f(\mathbf{p},\mathbf{t})=h(\mathbf{p},\mathbf{t})-g(\mathbf{p},\mathbf{t})
\end{equation}
where $h(\mathbf{p},\mathbf{t})\triangleq\sum_{k=1}^{K}w_{k}\sum_{n=1}^{N}\rho_{kn}(p_{kn}^{2}+t_{kn}^{2})$
and $g(\mathbf{p},\mathbf{t})\triangleq\sum_{k=1}^{K}w_{k}\sum_{n=1}^{N}\rho_{kn}(p_{kn}^{2}+t_{kn}^{2})+\log\Bigl(1+\frac{p_{kn}}{t_{kn}}\Bigr)$
are both convex. The values of $\rho_{kn}$'s are determined based
on Lemma \ref{lem:lemma1}. In light of the SCA-based approach, a
convex upper bound of $f(\mathbf{p},\mathbf{t})$ is given by
\begin{multline}
f(\mathbf{p},\mathbf{t})\leq F(\mathbf{p},\mathbf{t};\tilde{\mathbf{p}},\tilde{\mathbf{t}})\triangleq h(\mathbf{p},\mathbf{t})-g(\tilde{\mathbf{p}},\tilde{\mathbf{t}})\\
-\bigl\langle\bigl[\nabla_{\mathbf{p}}g(\tilde{\mathbf{p}},\tilde{\mathbf{t}});\nabla_{\mathbf{t}}g(\tilde{\mathbf{p}},\tilde{\mathbf{t}})\bigr]\trans,\bigl[\mathbf{p}-\tilde{\mathbf{p}};\mathbf{t}-\tilde{\mathbf{t}}\bigr]\bigr\rangle\label{eq:logapp}
\end{multline}
which is a quadratic function. Consequently, the problem considered
in the $\theta$th iteration of the SCA-based approach is given by
\begin{subequations}\label{eq:mumc:convex}
\begin{align}
\underset{\mathbf{p},\mathbf{t}}{\mini} & \;\tilde{f}(\mathbf{p},\mathbf{t})\triangleq F(\mathbf{p},\mathbf{t};\mathbf{p}^{(\theta-1)},\mathbf{t}^{(\theta-1)})\\
\st & \;t_{kn}=\tilde{\sigma}_{kn}^{2}+\sum_{l\neq k}\tilde{G}_{kln}p_{ln},\,\forall k,n\label{eq:mcslackvariable}\\
 & \;\sum_{n=1}^{N}p_{kn}\leq\bar{p}_{k},\,\forall k;\;0\leq p_{kn}\leq\bar{p}_{kn},\,\forall k,n.
\end{align}
\end{subequations} Now we can substitute $t_{kn}$ by $\tilde{\sigma}_{kn}^{2}+\sum_{l\neq k}\tilde{G}_{kln}p_{ln}$
into the objective and omit \eqref{eq:mcslackvariable} to achieve
an equivalent conic formulation in which the optimization variable
is only $\mathbf{p}$. We remark that, in reality, the parameter $\rho_{kn}$
can be made smaller than the bound given above which in most cases
can accelerate the convergence of the SCA procedure.

\subsubsection{Proposed Solution for MaxminEEfair}

A solution to \eqref{eq:mumc:EEfair} can be obtained similarly. To
this end we first rewrite \eqref{eq:mumc:EEfair} as \begin{subequations}\label{eq:mumc:EEfair:reform}
\begin{align}
\underset{\mathbf{p},\mathbf{t},\vartheta>0}{\maxi} & \;\vartheta\\
\st & \;\frac{\sum_{n=1}^{N}t_{kn}}{\sum_{n=1}^{N}p_{kn}+p_{k}^{\mathrm{c}}}\geq\vartheta,\forall k\label{eq:mumceenonconvex}\\
 & \;u_{kn}\log\frac{u_{kn}}{\tilde{u}_{kn}}\geq u_{kn}t_{kn},\,\forall k,n\label{eq:mumceenoconvex2}\\
 & \;\sum_{n=1}^{N}p_{kn}\leq\bar{p}_{k},\,\forall k;\;0\leq p_{kn}\leq\bar{p}_{kn},\,\forall k,n.
\end{align}
\end{subequations} where $u_{kn}$ and $\tilde{u}_{kn}$ are defined
below \eqref{eq:PCWSR:gen}. The nonconvex constraints in \eqref{eq:mumceenonconvex}
can be approximated by \eqref{eq:App1}, and those in \eqref{eq:mumceenoconvex2}
can be handled in the same way as done for \eqref{eq:PCWSR:softSINR-1}.

The worst-case complexity of solving the SOCP problems in \eqref{eq:PCWSR:DC}
and \eqref{eq:mumc:EEfair:reform} by interior-point methods is $\mathcal{O}\left(K^{3.5}N{}^{3.5}\right)$.
On the other hand, the worst-case complexity estimate for the method
in \cite{Wang:SCALE:2012} (for solving the WSRmax problem) and the
one in \cite{Zappone:EEpowercontrol:2015} (for solving the maxminEEfair
problem) is $\mathcal{O}\left(K^{6}N^{6}\right)$\cite[chapter 5]{lecture_on_modernCO}.
This comparison shows a huge improvement of the proposed SOCP-based
approach over the existing solutions in terms of solution speed, which
is numerically confirmed in the following.

\subsection{Numerical Results}

To evaluate the solutions in this section, we adopt the simulation
model in \cite{Wang:SCALE:2012}. Specifically, the coordinates of
transmitter and receiver $k$ in meters are $(k,10)$ and $(k,0)$,
respectively. The noise variance at all subcarriers is $\sigma^{2}=-30$
dBm. The path loss attenuation is $\alpha d^{-3}$, where $\alpha$
is the log normal shadowing with the standard deviation of $3$ and
$d$ is the distance in meters. The multipath channels are modeled
with six taps (delay) which are circularly symmetric complex Gaussian
random variable with zero mean and variance vector as $[1\,e^{-3}\,e^{-6}\,e^{-9}\,e^{-12}\,e^{-15}]$.
The number of transmission links $K$ and the number of subcarriers
$N$ are specified in each experiment. The maximum transmit power
at the transmitters are $\bar{p}_{k}=32$ dBm, $\forall k$. To create
the initial points for the algorithms, we uniformly allocate power
to subcarriers such that \eqref{eq:mcpowerconstraint} is satisfied.

\subsubsection{Weighted Sum Rate Maximization}

\begin{figure}
\centering
\subfigure[Average weighted sum rate (in bits per second).]{\label{fig:sumratemcrate}\includegraphics[width=0.95\columnwidth]{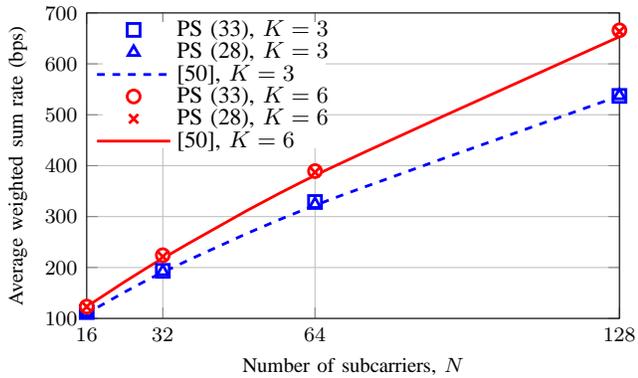}}\\
\subfigure[Average run time (in seconds).]{ \label{fig:sumratemctime}\includegraphics[width=0.95\columnwidth]{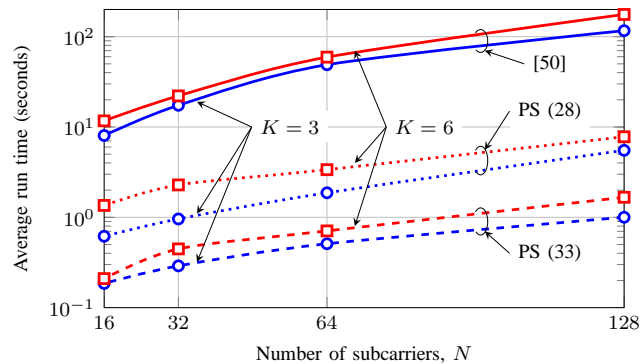}}

\caption{Average weighted sum rate and average run time of the proposed solutions
(PS \eqref{eq:PCWSR:gen} and PS \eqref{eq:mumc:convex}) and the
solution in \cite{Wang:SCALE:2012} for WSRmax problem with different
number of subcarrier $N$. The number of transmission links $K$ is
taken as 3 and 6. }
 \label{fig:sumratemcperf}
\end{figure}

We first compare the proposed solutions (PS \eqref{eq:PCWSR:gen}
and PS \eqref{eq:mumc:convex}) with the algorithm in \cite{Wang:SCALE:2012}
for the WSRmax problem. The iterative procedure in \cite{Wang:SCALE:2012}
stops when the increase in the weighted sum rate of 5 consecutive
iterations is less than $10^{-3}.$ The average weighted sum rate
performance of the considered methods are plotted in Fig. \ref{fig:sumratemcrate}
for different numbers of $N$ and $K$. As can be seen, these approaches
achieve nearly the same performance in all cases of $K$ and $N$.
The complexity comparison of the methods is shown in Fig. \ref{fig:sumratemctime}.
Clearly, the run time of the GP-based method in \cite{Wang:SCALE:2012}
scales very fast with $N$, compared to that of the SOCP-based methods.
Thus, the GP-based method requires prohibitively high computation
time for large $N$. In particular, when $(K,N)=(3,128)$, the proposed
solutions \eqref{eq:PCWSR:gen} and \eqref{eq:mumc:convex} are approximately
20 times and 100 times faster than the method in \cite{Wang:SCALE:2012},
respectively. Another expected result is that the proposed solution
\eqref{eq:mumc:convex} achieves better computational efficiency compared
to the proposed solution \eqref{eq:PCWSR:gen}. This gain comes from
the facts that the number of variables in \eqref{eq:PCWSR:gen} is
larger than that of \eqref{eq:mumc:convex}, and \eqref{eq:mumc:convex}
is a QP.

\subsubsection{Max-min Fairness Energy Efficiency}

\begin{figure}
\centering
\subfigure[Average minimum energy efficiency (in bits/mJ).]{\label{fig:eemcee}\includegraphics[width=0.95\columnwidth]{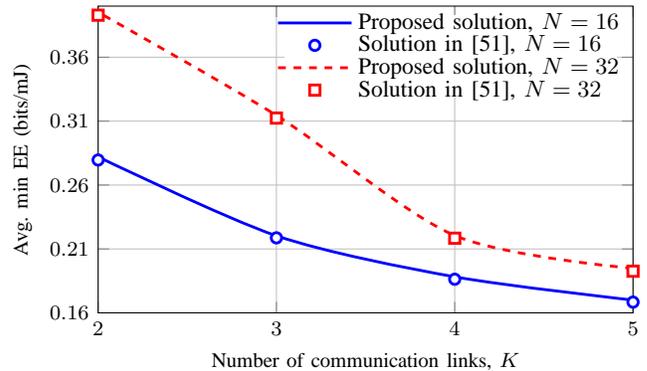}}\\
\subfigure[Average run time (in seconds).]{ \label{fig:eemctime}\includegraphics[width=0.95\columnwidth]{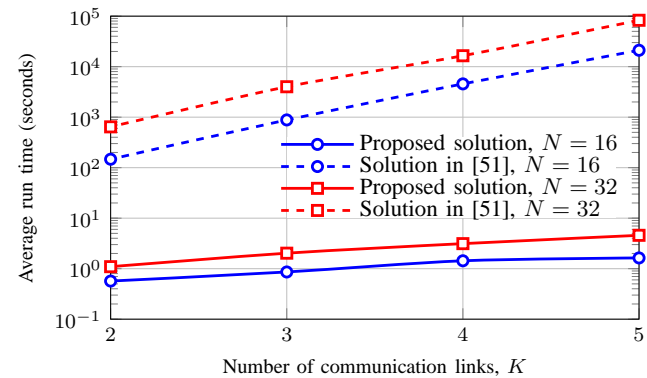}}

\caption{Average minimum energy efficiency and average run time of the proposed
solution and the solution in \cite{Zappone:EEpowercontrol:2015} for
maxminEEfair problem with different sets of $(K,N)$. The circuit
power $p_{k}^{\mathrm{c}}$ is $5$ dBm for all $k$.}
 \label{fig:eemcperf}
\end{figure}

For maxminEEfair problem, we compare the proposed solution with the
two-stage iterative method in \cite{Zappone:EEpowercontrol:2015}.
For a fair comparison, the stopping criterion of the two-stage iterative
algorithm is as follows. The tolerance error of the inner-stage (i.e.
Dinkelbach's procedure) is $10^{-3}$, and the outer-stage (i.e. the
SCA loop) stops when the increase of 5 consecutive iterations is less
than $10^{-3}.$ The solver for each subproblem of the two-stage iterative
algorithm is FMINCON, which is the general nonlinear solver included
in MATLAB's Optimization Toolbox.

The achieved minimum energy efficiency of the two methods in comparison
is shown in Fig. \ref{fig:eemcee}. In all cases of $(K,N)$, the
energy efficiency performance of the two approaches is the same, but
there is a huge difference in terms of computation time as shown in
Fig. \ref{fig:eemctime}. Again, we can see that the run time of the
proposed solution is much less sensitive to the dimension of the problems,
compared to the existing method. In particular, the proposed solution
is about $10^{4}$ times faster than the method in \cite{Zappone:EEpowercontrol:2015}
for $K=5$ and $N=32$.

\section{Conclusion and Discussion \label{sec:Conclusion}}

We have proposed several conic quadratic approximations for wireless
communications designs in the context of the SCA paradigm and applied
these to solve various design problems, including AF beamforming,
cognitive multicasting, MIMO relaying, and multicarrier power controlling.
For some specific problems, modifications and customization have been
made to improve the solution efficiency. Numerical results have shown
that the proposed approximations are far superior to the existing
methods in terms of solution speed, while still achieve the same design
objective.

For future work, it is interesting to investigate the global optimality
of the computed solutions of different approximations proposed in
this paper. Another research direction can be finding a flexible and
efficient way to switch between different approximations during the
iterative process to speed up the convergence. \bibliographystyle{IEEEtran}

\end{document}